\documentclass[11pt]{article}
\usepackage{etex}
\usepackage[paper=letterpaper,margin=1in]{geometry}
\usepackage{graphics}
\usepackage{epsfig}
\usepackage{arcs}
\usepackage{fontenc}
\usepackage{alltt}
\usepackage{amssymb,amsfonts,epsf,url}
\usepackage{graphicx}
\usepackage{pstricks,pstricks-add}
\usepackage{pst-node}
\usepackage{pst-coil}
\usepackage{amsmath}
\usepackage{amsthm}
\usepackage{verbatim}

\usepackage{enumerate}
\usepackage{tikz}
\usepackage{amssymb,amsfonts}
\usepackage{amsmath,comment}

\usepackage{fancyhdr}

\newtheorem{theorem}{Theorem}[section]
\newtheorem{lemma}[theorem]{Lemma}

\newtheorem{proposition}[theorem]{Proposition}

\newtheorem{definition}[theorem]{Definition}

\newlength{\alginputwidth}
\newlength{\algboxwidth}

\newsavebox{\algbox}
\newsavebox{\captionbox}
    {
        \setlength{\algboxwidth}{\columnwidth}
        \addtolength{\algboxwidth}{-\columnsep}
        \addtolength{\algboxwidth}{-1mm}
        \setlength{\alginputwidth}{\algboxwidth}
        \addtolength{\alginputwidth}{-1.7cm}
        \begin{figure}[tb]
            \vspace*{2mm}
            \centering
            \begin{lrbox}{\captionbox}
                \begin{minipage}[b]{\algboxwidth}
                    \centering
                    \caption{#1}
                    \label{#2}
                \end{minipage}
            \end{lrbox}
            \begin{lrbox}{\algbox}
                \begin{minipage}[b]{\algboxwidth}
                    \footnotesize
                    \vspace*{2mm}
    } 
    {
                    \vspace*{0.2mm}
               \end{minipage}
            \end{lrbox}
            \fbox{\usebox{\algbox}\hspace*{1mm}}
            \usebox{\captionbox}
            \vspace*{-4mm}
        \end{figure}
    }
\newsavebox{\algcodebox}
    {
        \begin{enumerate}
            \setlength{\itemsep}{2pt}
            \setlength{\parsep}{0pt}
            \setlength{\topsep}{0pt}
            \setlength{\parskip}{0pt}
            \setlength{\partopsep}{0pt}
    } 
    {\end{enumerate}}

\newcommand{\paramproblem}[4]{\noindent {\sc #1}
\\
{\bf Given:} #2\\
{\bf Parameter:} #3\\
{\bf Question:} #4}

\newcommand{\Oh}{{\mathcal O}}
\newcommand{\Ohstar}{{\mathcal O^{*}}}
\newcommand{\nat}{\mathbb{N}}

\newcommand{\Pol}{\mbox{$\mathcal P$}}
\newcommand{\NP}{\mbox{$\mathcal{NP}$}}
\newcommand{\APX}{\mbox{$\mathcal{APX}$}}
\newcommand{\FPT}{\text{$\mathcal{FPT}$}}

\let\phi=\varphi
\let\epsilon=\varepsilon

\def\eg{{\em e.g.}}

\psset{unit=1pt}

\date{}
\title{On the Ordered List Subgraph Embedding Problems\footnote{A preliminary version of this paper appeared in proceedings of the {\em 8th International Symposium on Parameterized and Exact Computation} (IPEC), volume 8246 of {\em Lecture Notes in Computer Science}, pages~189--201, 2013.}}

\author{
{\sc Olawale Hassan}\footnote{School of
Computing, DePaul University, 243 S. Wabash Avenue, Chicago, IL
60604. Email: {\tt oahassan@gmail.com}, {\tt ikanj@cs.depaul.edu}, {\tt lperkovic@cs.depaul.edu}}
\and
{\sc Iyad Kanj}\footnotemark[2]
\and
{\sc Daniel Lokshtanov}\thanks{Department of Informatics,
University of Bergen, Bergen, Norway.  Email: {\tt daniello@ii.uib.no}}
\and
{\sc Ljubomir Perkovi\'{c}}\footnotemark[2]
}

\begin{document}

\maketitle

\begin{abstract}
In the (parameterized) {\sc Ordered List Subgraph Embedding} problem (p-OLSE) we are given two graphs $G$ and $H$, each with a linear order defined on its vertices, a function $L$ that associates with every vertex in $G$ a list of vertices in $H$, and a parameter $k$. The question is to decide if we can embed (one-to-one) a subgraph $S$ of $G$ of $k$ vertices into $H$ such that: (1) every vertex of $S$ is mapped to a vertex from its associated list, (2) the linear orders inherited by $S$ and its image under the embedding are respected, and (3) if there is an edge between two vertices in $S$ then there is an edge between their images. If we require the subgraph $S$ to be embedded as an induced subgraph, we obtain the {\sc Ordered List Induced Subgraph Embedding} problem (p-OLISE). The p-OLSE and p-OLISE problems model various problems in Bioinformatics related to structural comparison/alignment of proteins.

We investigate the complexity of p-OLSE and p-OLISE with respect to the following structural parameters: the {\em width} $\Delta_L$ of the function $L$ (size of the largest list), and the maximum degree $\Delta_H$ of $H$ and $\Delta_G$ of $G$. In terms of the structural parameters under consideration, we draw a complete complexity landscape of p-OLSE and p-OLISE (and their optimization versions) with respect to the computational frameworks of classical complexity, parameterized complexity, and approximation.
\end{abstract}

\pagenumbering{roman}
\pagenumbering{arabic}

\section{Introduction} \label{sec:intro}
\subsection{Problem Definition and Motivation}
\label{subsec:motivation}
We consider the following problem that we refer to as the parameterized {\sc Ordered List Subgraph Embedding} problem, shortly (p-OLSE):

\paramproblem{} {Two graphs $G$ and $H$ with linear orders $\prec_G$ and $\prec_{H}$ defined on the vertices of $G$ and $H$; a function $L : V(G) \longrightarrow 2^{V(H)}$; and $k \in \nat$}{$k$}{Is there a subgraph $S$ of $G$ of $k$ vertices and an injective map $\phi: V(S) \longrightarrow V(H)$ such that: (1) $\phi(u) \in L(u)$ for every $u \in S$; (2) for every $u, u' \in S$, if $u \prec_G u'$ then $\phi(u) \prec_H \phi(u')$; and (3) for every $u, u' \in S$, if $uu' \in E(G)$ then $\phi(u)\phi(u') \in E(H)$} \\

The parameterized {\sc Ordered List Induced Subgraph Embedding} (p-{\sc OLISE}) problem, in which we require the subgraph $S$ to be embedded as an induced subgraph, is defined as follows:

\paramproblem{} {Two graphs $G$ and $H$ with linear orders $\prec_G$ and $\prec_{H}$ defined on the vertices of $G$ and $H$; a function $L : V(G) \longrightarrow 2^{V(H)}$; and $k \in \nat$}{$k$}{Is there a subgraph $S$ of $G$ of $k$ vertices and an injective map $\phi: V(S) \longrightarrow V(H)$ such that: (1) $\phi(u) \in L(u)$ for every $u \in S$; (2) for every $u, u' \in S$, if $u \prec_G u'$ then $\phi(u) \prec_H \phi(u')$; and (3) for every $u, u' \in S$, $uu' \in E(G)$ if and only if $\phi(u)\phi(u') \in E(H)$} \\

The optimization version of p-OLSE (resp. p-OLISE), denoted opt-OLSE (resp. opt-OLISE), asks for a subgraph $S$ of $G$ with the maximum number of vertices such that there exists a valid list embedding $\phi$ that embeds $S$ as a subgraph (resp. induced subgraph) into $H$.

The p-OLSE and p-OLISE problems, and their optimization versions opt-OLSE and opt-OLISE, have applications in the area of Bioinformatics because they model numerous protein and DNA structural comparison problems (see~\cite{xiuzhen,evansphd,evans,xiuzhencai}).

In this paper we investigate the complexity of p-OLSE and p-OLISE with respect to the following structural parameters: the {\em width} $\Delta_L$ of the function $L$ defined to be $max\{|L(u)| : u \in V(G)\}$ and the maximum degree $\Delta_H$ of $H$ and $\Delta_G$ of $G$. Restrictions on the structural parameters $\Delta_H$, $\Delta_G$ and $\Delta_L$ are very natural in Bioinformatics. The parameters $\Delta_H$ and $\Delta_G$ model the maximum number of hydrophobic bonds that an amino acid in each protein can have; on the other hand, $\Delta_L$ is usually a parameter set by the Bioinformatics practitioners when computing the top few alignments of two proteins~\cite{xiuzhencai}.

\subsection{Previous Related Results}\label{subsec:previous}
Goldman et al.~\cite{Goldman1999} studied protein comparison problems using the notion of {\em contact maps}, which are undirected graphs whose vertices are linearly ordered. They formulated the protein comparison problem as a {\sc Contact Map Overlap} problem, in which we are given two contact maps and we need to identify a subset of vertices $S$ in the first contact map, a subset of vertices $S'$ in the second with $|S| = |S'|$, and an order-preserving bijection $f: S \longrightarrow S'$, such that the number of edges in $S$ that correspond to edges in $S'$ is maximized. In~\cite{Goldman1999}, the authors proved that the {\sc Contact Map Overlap} problem is MAXSNP-complete, even when both contact maps have maximum degree one. The authors develop approximation algorithms by restricting their attention to special cases of the {\sc Contact Map Overlap} problem. In particular, they focus on cases where at least one of the graphs is a self-avoiding walk on the two dimensional grid.  Such walks are formulated as a one-to-one mapping $g : \{1, 2,\dotsc, n\} \to \mathbb{Z}\times\mathbb{Z}$ such that $||g(i)-g(i+1)||_2 = 1$ for $i = 1,\dotsc,n-1$.  Goldman et al.~\cite{Goldman1999} observed that the geometry of proteins exhibit this structure, making self-avoiding walks a natural restriction for comparing protein contact maps.  They create contact maps from walks by taking the set $\{1,2,\dotsc,n\}$ as the set of vertices and including an edge for every pair $ij$ such that $|i - j| > 1$ and $||g(i)-g(j)||_2 = 1$. While the problem remains \NP-complete even when comparing two contact maps that are self-avoiding walks, the authors give a polynomial-time $\frac{1}{4}$-approximation algorithm when comparing two self-avoiding walks and a polynomial-time $\frac{1}{3}$-approximation algorithm when comparing a self-avoiding walk and an arbitrary contact map. The main difference between the {\sc Contact Map Overlap} problem and the opt-OLSE and opt-OLISE problems under consideration is that in opt-OLSE and opt-OLISE the bijection $f$ is restricted to mapping a vertex to one in its list, and the goal is to maximize the size of the subgraph, not the number of edges, that can be embedded.

The p-OLISE problem generalizes the well-studied problem {\sc Longest Arc-Preserving Common Subsequence} (LAPCS) (see~\cite{gramm,evansphd,evans,guo,guohuidam,guohui}).  In LAPCS, we are given two sequences $S_1$ and $S_2$ over a fixed alphabet, where each sequence has arcs/edges between its characters, and the problem is to compute a longest common subsequence of $S_1$ and $S_2$ that respects the arcs. The LAPCS problem was introduced by~\cite{evansphd,evans} where it was shown to be $W[1]$-complete (parameterized by the length of the common subsequence sought) in the case when the arcs are {\em crossing}. Several works studied the complexity and approximation of LAPCS with respect to various restrictions on the types of the arcs (\eg, {\em nested}, {\em crossing}, etc.)~\cite{gramm,evansphd,evans,guo,guohuidam,guohui}. The work in~\cite{gramm,guo} considered the problem in the case of nested arcs parameterized by the total number of characters that need to be deleted from $S_1$ and $S_2$ to obtain the arc-preserving common subsequence. They showed that the problem is $\FPT$ with respect to this parameterization, and they also showed it to be $\FPT$ when parameterized by the length of the common subsequence in the case when the alphabet consists of four characters. The p-OLISE problem generalizes LAPCS since no restriction is placed on the size of the alphabet, and a vertex can be mapped to any vertex from its list. Consequently, the ``positive" results obtained in this paper about p-OLISE and opt-OLISE apply directly to their corresponding versions of LAPCS; on the other hand, we are able to borrow the $W[1]$-hardness result from~\cite{evans} to conclude the $W[1]$-hardness results in Proposition~\ref{prop:whardness} and Proposition~\ref{prop:whardness2}.

A slight variation of p-OLSE was considered in~\cite{xiuzhen,xiuzhencai}, where the linear order imposed on $G$ and $H$ was replaced with a partial order (directed acyclic graphs); the problem was referred to as the {\sc Graph Embedding} problem in~\cite{xiuzhen} and as the {\sc Generalized Subgraph Isomorphism} problem in~\cite{xiuzhencai}. The aforementioned problems were mainly studied in~\cite{xiuzhen,xiuzhencai} assuming no bound on $\Delta_H$ and $\Delta_G$ and, not surprisingly, only hardness results were derived. In~\cite{xiuzhencai}, a parameterized algorithm with respect to the treewidth of $G$ and the map width $\Delta_L$ combined was given. Most of the hardness results in~\cite{xiuzhen,xiuzhencai} were obtained by a direct reduction from the {\sc Independent Set} or {\sc Clique} problems. For example, it was shown in~\cite{xiuzhen} that the problem of embedding the whole graph $G$ into $H$ is \NP-hard, but is in $\Pol$ if $\Delta_L=2$. It was also shown that the problem of embedding a subgraph of $G$ of size $k$ into $H$ is $W[1]$-complete even when $\Delta_L =1$, and cannot be approximated to a ratio $n^{\frac{1}{2} - \epsilon}$ unless $\Pol=\NP$; we borrow these two hardness results as they also work for p-OLSE and p-OLISE.

Fagnot et al.~\cite{Fagnot2008178} introduced {\textsc Exact-$(\mu_G, \mu_H)$-Matching}, a problem in which we are given two graphs, $G$ and $H$, a mapping $L$ as in p-OLISE and p-OLSE, and constants $\mu_G$ and $\mu_H$ where max$\{|L(u)|:u \in G\} \le \mu_G$ and max$\{|L^{-1}(v)|:v \in H\} \le \mu_H$. However,  there is no ordering imposed on the vertices of $G$ or $H$.  The objective is to find an injective homomorphism from $G$ to $H$. The authors proved that {\textsc Exact-$(\mu_G, \mu_H)$-Matching} is \NP-Complete when $\mu_G \ge 3$ and $\mu_H = 1$, even if $G$ and $H$ are bipartite graphs.  Fertin et al. in~\cite{Fertin2005,Fertin2009} revisited the original problem, restricting their attention to graphs of bounded degree, and considered an optimization version, {\textsc Max-$(\mu_G, \mu_H)$-Matching}, where the objective is to embed as many edges from $G$ into $H$ as possible. While the authors found that for small values of $\Delta_G$ and $\Delta_H$ there are cases where {\textsc Exact-$(\mu_G, \mu_H)$-Matching} is in $\Pol$ and cases where {\textsc Max-$(\mu_G, \mu_H)$-Matching} has a constant ratio approximation, generally the problem remains difficult. Fertin et al.~\cite{Fertin2005,Fertin2009} also give an $\FPT$ algorithm for solving {\textsc Max-$(\mu_G, \mu_H)$-Matching} by adding an arbitrary ordering to $G$ and $H$ which indicates that ordering $G$ and $H$ might improve the tractability of graph embedding problems. These results contrast with our results in Section~\ref{sec:complexity} and Section~\ref{subsec:parameterized} that show that p-OLSE and p-OLISE become easy when the ordering constraint is removed.

Finally, one can draw some similarities between p-OLISE and the {\sc Subgraph Isomorphism} and {\sc Graph Embedding} problems. The main differences between p-OLISE and the aforementioned problems are: (1) in p-OLSE we have linear orders on $G$ and $H$ that need to be respected by the map sought, (2) we ask for an embedding of a subgraph of $G$ rather than the whole graph $G$, and (3) each vertex must be mapped to a vertex from its list. In particular, requirement (1) above precludes the application of well-known (logic) meta-theorems (see~\cite{grohebook}) to the restrictions of p-OLISE that are under consideration in this paper.

\subsection{Our Results and Techniques}\label{subsec:results}
We draw a complete complexity landscape of p-OLSE and p-OLISE with respect to the computational frameworks of classical complexity, parameterized complexity, and approximation, in terms of the structural parameters $\Delta_H$, $\Delta_G$ and $\Delta_L$. Table~1 outlines the obtained results about p-OLSE and p-OLISE and their optimization versions. Even though our hardness results are for specific values of the parameters $\Delta_H$, $\Delta_G$, and $\Delta_L$, these results certainly hold true for restrictions of the problems to instances in which the corresponding parameters are upper bounded by (or equal to --- by adding dummy vertices) any constant larger than these specific values. Observe also that the results we obtain {\em completely and tightly} characterize the complexity (with respect to all frameworks under consideration) of the problems with respect to $\Delta_H$, $\Delta_G$ and $\Delta_L$ (unbounded vs. bounded, and when applicable, for different specific values).

Section~\ref{sec:complexity} presents various classical complexity results. The \NP-hardness results are obtained by a reduction from the \textsc{$k$-Multi-Colored Independent Set} problem. Section~\ref{subsec:approximation} presents various approximation results for the optimization versions of p-OLSE and p-OLISE. Section~\ref{subsec:parameterized} presents parameterized complexity results for various restrictions of p-OLSE and p-OLISE. The $W[1]$-hardness results are obtained by tweaking the $W[1]$-hardness results given in the literature~\cite{xiuzhen,evans}, or by simple known reductions from the {\sc Independent Set} problem. The $\FPT$ results in Theorem~\ref{thm:mainboundeddegree} for p-OLSE, when $\Delta_L = O(1)$, $\Delta_G=O(1)$, and $\Delta_H = \infty$, are derived using the random separation method~\cite{cai}. This method is applied after transforming the problem --- via reduction operations --- to the {\sc Independent Set} problem on a graph composed of (1) a permutation graph and (2) a set of additional edges between the permutation graph vertices such that the number of additional edges incident to any vertex is at most a constant; Lemma~\ref{lem:separation} then shows that the {\sc Independent Set} problem on such graphs is $\FPT$. On the other hand, the $\FPT$ results in Proposition~\ref{prop:randomsimple}, when $\Delta_H=0$, $\Delta_G=O(1)$ (resp. $\Delta_G=0$ and $\Delta_H=O(1)$ for p-OLISE by symmetry) and $\Delta_L=\infty$, are also derived using the random separation method~\cite{cai}, but the argument is simpler. We also explore the contribution of the ordering constraint of p-OLSE and p-OLISE to the difficulty of the problems.

To cope with the $W$-hardness of p-OLSE in certain cases, we consider a different parameterization of the problem in Section~\ref{sec:vcnumber}, namely the parameterization by the vertex cover number, and denote the associated problem by p-VC-OLSE. This parameterization is not interesting for p-OLISE since we prove that p-OLISE is \NP-complete in the case when $\Delta_G=0$, $\Delta_H=1$ and $\Delta_L$ = 1, and hence the problem is {\em para-\NP-hard}~\cite{grohebook} with respect to this parameterization. Proposition~\ref{prop:whardnessvc} shows that p-VC-OLSE is $W[1]$-complete in the case when $\Delta_H=1$, $\Delta_G=1$ and $\Delta_L$ is unbounded (note that if either $\Delta_H=0$ or $\Delta_G=0$ then p-OLSE is $\FPT$ when $\Delta_L = \infty$). So we restrict our attention to the case when $\Delta_L=O(1)$, and show in this case that the problem is $\FPT$ even when both $\Delta_H$ and $\Delta_G$ are unbounded; the method relies on a bounded search tree approach, combined with the dynamic programming algorithm described in Proposition~\ref{prop:noedges}.

\begin{table}[htbp] \label{table:map}
\begin{center} \scriptsize
\mbox{\renewcommand{\arraystretch}{1.1}
\begin{tabular}{p{2.65cm}p{1cm}p{1cm}p{1cm}l}
{\small\bf p-OLSE}   & {\bf $\Delta_H$} & {\bf $\Delta_G$} & {\bf $\Delta_L$} & {\bf Complexity} \\ \hline
Classical     & $\infty$ & 0 & $\infty$ & $\Pol$ \\
              & $0$ & $1$ & $1$ & $\NP$-complete \\ \hline
Approximation & $\infty$ & $O(1)$ & $\infty$ & $\APX$-complete \\
              & $0$ & $\infty$ & $1$ &  not approximable to $n^{\frac{1}{2} - \epsilon}$  \\ \hline
Parameterized & $\infty$ & $O(1)$ & $O(1)$ &  $\FPT$ \\
              & $1$ & $1$ & $\infty$ & $W[1]$-complete \\
              & $0$ & $\infty$ & $1$ & $W[1]$-complete  \\
              & $0$ & $O(1)$ & $\infty$ & $\FPT$ \\
              & $\infty$ & $0$ & $\infty$ & $\FPT$ (even in $\Pol$) \\ \hline
\end{tabular}}

\vspace*{0.5cm}

\mbox{\renewcommand{\arraystretch}{1.1}
\begin{tabular}{p{2.65cm}p{1cm}p{1cm}p{1cm}l}
{\small\bf p-OLISE}   & {\bf $\Delta_H$} & {\bf $\Delta_G$} & {\bf $\Delta_L$} & {\bf Complexity} \\ \hline
Classical     & $0$ & $0$ & $\infty$ & $\Pol$ \\
              & $0$ & $1$ & $1$ & $\NP$-complete \\
              & $1$ & $0$ & $1$ & $\NP$-complete \\ \hline
Approximation & $O(1)$ & $O(1)$ & $\infty$ & $\APX$-complete \\
              & $ 0$ & $\infty$ & $1$ &  not approximable to $n^{\frac{1}{2} - \epsilon}$   \\
              & $\infty$ & $0$ & $1$ &  not approximable to $n^{\frac{1}{2} - \epsilon}$ \\ \hline
Parameterized & $\infty$ & $0$ & $1$ & $W[1]$-complete \\
              & $0$ & $\infty$ & $1$ & $W[1]$-complete \\
              & $0$ & $O(1)$ & $\infty$ &   $\FPT$ \\
              & $O(1)$ & $0$ & $\infty$ &   $\FPT$ \\
              & $1$ & $1$ & $1$ & $W[1]$-complete \\
\end{tabular}}

\caption{Classical, approximation, and parameterized complexity maps of p-OLSE and p-OLISE with respect to $\Delta_H$, $\Delta_G$ and $\Delta_L$. The inapproximability results are under the assumption that $\Pol \neq \NP$. The symbol $\infty$ stands for unbounded degree.}
\end{center} \vspace*{-0.5cm}
\end{table}

\subsection{Background and Terminologies}\label{sec:background}
\paragraph{Graphs.} For a graph $H$ we denote by $V(H)$ and $E(H)$ the set of vertices and edges of $H$, respectively; we write $|H|$ for $|V(H)|$.  For a set of vertices $S \subseteq V(H)$, we denote by $H[S]$ the subgraph of $H$ induced by the vertices in $S$. For a subset of edges $E' \subseteq E(H)$, we denote by $H-E'$ the graph $(V(H), E(H) \setminus E')$. For a vertex $v \in H$, $N(v)$ denotes the set of neighbors of $v$ in $H$. The {\em degree} of a vertex $v$ in $H$, denoted $deg_H(v)$, is $|N(v)|$.  A vertex $v$ is {\em isolated} in $H$ if $deg_H(v)=0$. The {\em degree} of $H$, denoted $\Delta(H)$, is $\Delta(H) = \max\{deg_H(v): v \in H\}$. A {\em matching} in a graph is a set of edges $M$ such that no two edges in $M$ share an endpoint. An {\em independent set} of a graph $H$ is a set of vertices $I$ such that no two vertices in $I$ are adjacent.  A {\em maximum independent set} in $H$ is an independent set of maximum cardinality.  A {\em clique} is subset of vertices $Q \subseteq V(H)$ such that for each pair of vertices $u, u' \in Q, uu' \in E(H)$.  A {\em vertex cover} of $H$ is a set of vertices such that each edge in $H$ is incident to at least one vertex in this set; we denote by $\tau(H)$ the cardinality of a minimum vertex cover of $H$. Let $L$ and $L'$ be two parallel lines in the plane. A {\em permutation graph} $P$ is the intersection graph of a set of line segments such that one endpoint of each of those segments lies on $L$ and the other endpoint lies on $L'$. For more information about graphs, we refer the reader to~\cite{west}.

\paragraph{Optimization Problems.} An {\it \NP optimization problem} $Q$ is a 4-tuple $(I_Q, S_Q, f_Q, g_Q)$. $I_Q$ is the set of input instances, which is recognizable in polynomial time. For each instance $x \in I_Q$, $S_Q(x)$ is the set of feasible solutions for $x$, which is defined by a polynomial $p$ and a polynomial-time computable predicate $\pi$ ($p$ and $\pi$ depend only on $Q$) as $S_Q(x) = \{ y : |y| \leq p(|x|) \;\wedge\; \pi(x, y)\}$. The function $f_Q(x, y)$ is the objective function mapping a pair $x \in I_Q$ and $y \in S_Q(x)$ to a non-negative integer.  The function $f_Q$ is computable in polynomial time. The function $g_Q$ is the {\em goal} function, which is one of the two functions $\{\max, \min\}$, and $Q$ is called a {\it maximization problem} if $g_Q = \max$, or a {\it minimization problem} if $g_Q = \min$. We will denote by $opt_Q(x)$ the value $g_Q\{f_Q(x, z) \mid z \in S_Q(x) \}$, and if there is no confusion about the underlying problem $Q$, we will write $opt(x)$ to denote $opt_Q(x)$.

\paragraph{Polynomial Time Approximation Algorithms.} A polynomial time algorithm $A$ is an {\it approximation algorithm} for a (maximization)
problem $Q$ if for each input instance $x \in I_Q$ the algorithm $A$
returns a feasible solution $y_A(x) \in S_Q(x)$.  The solution $y_A(x)$
has an {\it approximation ratio $r(|x|)$} if it satisfies the following
condition:
\begin{eqnarray*}
opt_Q(x)/f_Q(x, y_A(x)) \leq r(|x|).
\end{eqnarray*}

The approximation algorithm $A$ has an {\it approximation ratio $r(|x|)$} if
for every instance $x$ in $I_Q$ the solution $y_A(x)$ constructed by the
algorithm $A$ has an approximation ratio bounded by $r(|x|)$.

An optimization problem $Q$ has a {\em constant-ratio approximation algorithm} if it has an approximation algorithm whose ratio is a constant.

Given two maximization problems $Q = (I_Q, S_Q, f_Q, g_Q)$ and $Q' = (I_{Q'}, S_{Q'}, f_{Q'}, g_{Q'})$, we say that $Q$ {\em L-reduces} to $Q'$~\cite{py91} if there are two polynomial-time computable functions/algorithms $h, h'$ and constants $\alpha$ and $\beta$ such that:

\begin{enumerate}[(i)]
  \item For all $x \in I_Q$, $h$ produces an instance $x' = h(x)$ of $I_{Q'}$, such that $opt_{Q'}(x') \le \alpha \cdot opt_Q(x)$, and
  \item for any solution $y' \in S_{Q'}(x')$, $h'$ produces a solution $y \in S_{Q}(x)$ such that~$|f_Q(y) - opt_Q(x)| \le \beta \cdot |f_{Q'}(y') - opt_{Q'}(x')|$.
\end{enumerate}

\paragraph{Parameterized Complexity.} A {\it parameterized problem} is a set of instances of the form $(x,
k)$, where $x \in \Sigma^*$ for a finite alphabet set $\Sigma$, and
$k$ is a non-negative integer called the {\em parameter}.
A parameterized problem $Q$ is {\it fixed parameter tractable} ($\FPT$), if there exists an algorithm that on input $(x, k)$
decides if $(x, k)$ is a yes-instance of $Q$ in time $f(k)|x|^{O(1)}$,
where $f$ is a computable function independent of $|x|$; we will denote by {\em fpt-time} a running time of the form $f(k)|x|^{O(1)}$.
A hierarchy of fixed-parameter intractability, {\it the $W$-hierarchy} $\bigcup_{t
\geq 0} W[t]$, was introduced based on the notion of {\em fpt-reduction}, in which the $0$-th level $W[0]$ is the class $\FPT$.  It is commonly believed that $W[1] \neq \FPT$. The asymptotic notation $O^*()$ suppresses a polynomial factor in the input length. For more information about parameterized complexity, we refer the reader to~\cite{fptbook,grohebook,rolfbook}.

\paragraph{Randomized Algorithms.} A {\em randomized algorithm} is an algorithm that relies on random data as part of its input.  Such algorithms can be {\em de-randomized} using techniques to enumerate values of the randomized input, yielding a deterministic algorithm.  In this paper we design randomized algorithms using the color coding technique introduced by Alon et al.~\cite{alon95} and the random separation technique introduced by Cai et al.~\cite{cai}.  Both techniques are applicable to problems where the objective is to find a subset of size $k$ satisfying certain properties in a universe of size $n$.

Given a problem where we seek a subset of size $k$ satisfying certain properties in a universe of $n$ elements, color coding works by randomly assigning $k$ colors to the set of $n$ elements and then looking for a subset of size $k$ that is {\em colorful}, that is, where each element of the subset has a distinct color.  The algorithm can be de-randomized using families of $k$-perfect hash functions, or colorings, constructed as shown in~\cite{alon95} in time $\Oh(2^{\Oh(k)}\log(n))$. Thus, color coding algorithms run in time $\Oh(2^{\Oh(k)}\log(n)f(n,k))$ where $f(n, k)$ is the running time to verify that a colorful subset of size $k$ satisfying the desired properties exists for a given coloring.  If $f$ is polynomial in $n$, then the de-randomized algorithm runs in fpt-time.

Random separation~\cite{cai} works by randomly partitioning the vertices of a graph $G$ into red and green vertices and testing whether a solution $S \subseteq V(G)$ exists such that $S$ is entirely contained in the green partition and $N(S)$ is entirely contained in the red partition.  The likelihood that a solution $S$ is colored green and $N(S)$ is colored red is $2^{-(k+|N(S)|)}$.  To derandomize the algorithm we use $(n,t)$-universal sets where $t = k+|N(S)|$.  Naor et al.~\cite{naor} give a construction of $(n,t)$-universal sets that runs in time $\Oh(2^tt^{\Oh(\log(t))}\log(n))$.  Thus, the running time for the de-randomized algorithm is $\Oh(2^tt^{\Oh(\log(t))}\log(n)f(n,k))$, where $f(n,k)$ is the time it takes to verify that a solution exists in the green partition and its neighbors are in the red partition.  If $f$ is polynomial in $n$ and $t$ is bounded by a function of $k$ then the algorithm runs in fpt-time.

We will denote an instance of p-OLSE or p-OLISE by the tuple $(G, H, \prec_G, \prec_H, L, k)$. We shall call an injective map $\phi$ satisfying conditions (1)-(3) in the definition of p-OLSE and p-OLISE (given in Section~\ref{sec:intro}) for some subgraph $S$ of $G$, a {\em valid list embedding}, or simply a {\em valid embedding}. Constraint (3) will be referred to as the {\em embedding constraint} (note that constraint (3) is different in the two problems). We define the {\em width} of $L$, denoted $\Delta_L$, as $\max\{|L(u)| : u \in G\}$. It is often more convenient to view/represent the map $L$ as a set of edges joining every vertex $u \in G$ to the vertices of $H$ that are in $L(u)$.

\section{Classical Complexity Results}\label{sec:complexity}
In this section we consider p-OLSE and p-OLISE, and their optimization versions, with respect to classical complexity. Consider the restrictions of the opt-OLSE and opt-OLISE problems to instances in which $\Delta_G = \Delta_H=0$ (for p-OLSE we can even assume that $\Delta_H = \infty$ as the edges in $H$ do not play any role when $\Delta_G=0$). This version of the problem can be easily shown to be solvable in polynomial time by dynamic programming. The algorithm is very similar to that of the {\sc Longest Common Subsequence} problem; so we only give the recursive definition of the solution necessary for the dynamic programming approach.

\begin{proposition}\label{prop:noedges}
The opt-OLSE problem restricted to instances in which $\Delta_G=0$ and  $\Delta_H=\infty$, and the opt-OLISE problem restricted to instances in which $\Delta_G= \Delta_H=0$ are solvable in $O(|V(G)| \cdot |V(H)|)$ time (and hence are in $\Pol$).
\end{proposition}

\begin{proof}
Assume that the vertices of $G$ are ordered as $u_1, \ldots, u_n$ with respect to $\prec_G$, and that the vertices of $H$ are ordered as $v_1, \ldots v_N$ with respect to $\prec_H$.
We maintain a two-dimensional table $T$, in which $T[i, j]$ is the maximum cardinality of a subset of $\{u_1, \ldots, u_i\}$ from which there exists a valid list embedding into $\{v_1, \ldots, v_j\}$. The entry $T[i, j]$ can be computed recursively as follows: $T[i, j] = 1+ T[i-1, j-1]$ if $v_j \in L(u_i)$, and $T[i, j] = \max\{T[i, j-1], T[i-1, j]\}$ otherwise. We can now use a standard dynamic programming approach to compute $T$ and to construct the subgraph $S$ and the embedding $\phi$. The running time of the algorithm is $O(|V(G)| \cdot |V(H)|)$.
\end{proof}

Proposition~\ref{prop:noedges} will be useful for Proposition~\ref{prop:apx} and Theorem~\ref{thm:fptvslose}.

As we show next, if $\Delta_G > 0$, the p-OLSE and p-OLISE problems become $\NP$-complete, even in the simplest case when $\Delta_G=1$, $\Delta_H=0$ and $\Delta_L =1$. The same proof shows the \NP-completeness
of p-OLISE when $\Delta_H=1$, $\Delta_G=0$ and $\Delta_L =1$ by symmetry.

\begin{theorem}
The p-OLSE and p-OLISE problems restricted to instances in which $\Delta_H =0$, $\Delta_G = 1$, and $\Delta_L = 1$ are $\NP$-complete.
\end{theorem}

\begin{proof}
We first present the proof for p-OLSE.

It is easy to see that p-\textsc{OLSE} $\in \NP$, so it suffices to show that
p-\textsc{OLSE} is $\NP$-hard.  We do so by providing a polynomial time
reduction from the \textsc{$k$-Multi-Colored Independent Set ($k$-MCIS)} problem:
decide whether for a given graph $M = (V(M),E(M))$ and a proper $k$-coloring of the
vertices $f : V(M) \longrightarrow C$, where $C = \{1,2,..., k\}$ and each color class has the same cardinality, there exists an independent
set $I \subseteq V(M)$ of size $k$ such that, $\forall u, v \in I$, $f(u) \ne f(v)$~\cite{Hartung2013}.

Let $(M, f)$ be an instance of \textsc{$k$-MCIS}.
We assume that the set $C_i$, $i=1, \ldots, k$, of vertices that are mapped to color $i \in C$
is of size $N$, and we label the vertices of $M$ $u_1, \dots, u_{kN}$ such that
$C_i = \{u_{(i-1)N + 1}, \dots, u_{iN}\}$ (see Figure~1). We
describe how we construct the corresponding instance $(G, H, \prec_G,
\prec_H, L, k)$ of p-\textsc{OLSE}.

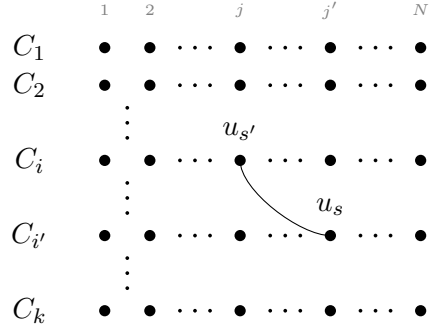
\begin{figure}
\label{fi:kMCIS}
\begin{center}
\begin{tikzpicture}
[vertex/.style={fill,circle,inner sep = 1.5pt},
dot/.style={fill,circle,inner sep = 0.5pt}]

\node at (0,4) [color=gray] {\tiny$1$};
\node at (0.6,4) [color=gray] {\tiny$2$};
\node at (1.8,4) [color=gray] {\tiny$j$};
\node at (3,4) [color=gray] {\tiny$j'$};
\node at (4.2,4) [color=gray] {\tiny$N$};

\foreach \y in {0,1,2,3,3.5}
{
  \foreach \x in {0,0.6,1.8, 3,4.2}
  {
    \node at (\x, \y) [vertex] {};
  }
  \draw (1,\y) node[dot] {};
  \draw (1.2,\y) node[dot] {};
  \draw (1.4,\y) node[dot] {};
  \draw (2.2,\y) node[dot] {};
  \draw (2.4,\y) node[dot] {};
  \draw (2.6,\y) node[dot] {};
  \draw (3.4,\y) node[dot] {};
  \draw (3.6,\y) node[dot] {};
  \draw (3.8,\y) node[dot] {};
}
\node at (-1,3.5) {$C_1$};
\node at (-1,3) {$C_2$};
  \draw (0.3,2.7) node[dot] {};
  \draw (0.3,2.5) node[dot] {};
  \draw (0.3,2.3) node[dot] {};
\node at (-1,2) {$C_i$};
  \draw (0.3,1.7) node[dot] {};
  \draw (0.3,1.5) node[dot] {};
  \draw (0.3,1.3) node[dot] {};
\node at (-1,1) {$C_{i'}$};
  \draw (0.3,.7) node[dot] {};
  \draw (0.3,.5) node[dot] {};
  \draw (0.3,.3) node[dot] {};
\node at (-1,0) {$C_k$};

\draw (1.8,2) node[label=90:{$u_{s'}$}] {};
\draw (3,1) node[label=90:{$u_s$}] {};
\draw (1.8,2) .. controls +(-90:0.4cm) and +(180:0.4cm) .. (3,1);
\end{tikzpicture}
\end{center}
\caption{An instance of \textsc{$k$-MCIS} consists of $kN$ vertices partitioned
into $k$ color classes $C_1, \dots, C_k$ each of size $N$. All edges are
between vertices in different color classes. The edge $(u_{s'}, u_s)$, for
example, connects vertices $u_{s'}$ and $u_s$ where $s' = (i-1)N+j$ and
$s=(i'-1)N+j'$.}
\end{figure}

To construct $G = (V(G),E(G))$, we associate with every color class $C_i$ a
block $B_i$ of vertex sequences $b_i^1,\dots,b_i^N$ such that each $b_i^j$, $j=1, \ldots, N$,
is a sequence of $kN$ vertices $u_{i,1}^j,\dots,u_{i,kN}^j$. $V(G)$ is then
the set of all vertices $u_{i,s}^j$ for $i = 1,\dots,k$, $j=1,\dots,N$, and
$s = 1,\dots,kN$. The ordering $\prec_G$ of vertices in $G$ is defined as
follows: for all $u_{i,s}^j, u_{i',s'}^{j'} \in G$
$$
u_{i,s}^j \prec_G u_{i',s'}^{j'} \iff
i < i' \mbox{ or } (i = i' \mbox{ and } j < j') \mbox{, or }
(i=i' \mbox{ and } j=j' \mbox{ and } s < s').
$$
(See Figure~2-(a).) Lastly we define the set $E(G)$ as follows: $(u_{i,s}^j, u_{i',s'}^{j'})$
is an edge in $E(G)$ if and only if $i < i'$, $s = (i'-1)N+j'$, $s'=(i-1)N+j$,
and $(u_s, u_{s'}) \in E(M)$. (See Figure~2-(b).)

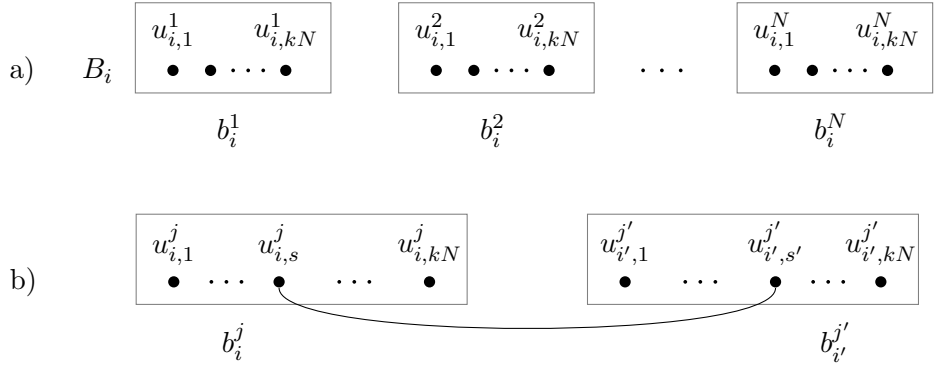
\begin{figure}
\label{fig100}
\begin{tikzpicture}
[vertex/.style={fill,circle,inner sep = 1.5pt},
dot/.style={fill,circle,inner sep = 0.5pt}]

\node at (-2, 0) {a)};
\node at (-1, 0) {$B_i$};

\foreach \x in {0,3.5,8}
{
  \draw [color=gray] (\x-0.5,-0.3) rectangle (\x+2.1,0.9);
  \foreach \e in {0, 0.5, 1.5}
  {
    \node at (\x+\e, 0) [vertex] {};
  }
  \draw (\x+0.8,0) node[dot] {};
  \draw (\x+1,0) node[dot] {};
  \draw (\x+1.2,0) node[dot] {};
}

\draw (6.25,0) node[dot] {};
\draw (6.5,0) node[dot] {};
\draw (6.75,0) node[dot] {};

\node at (0,0) [label=90:{$u^1_{i,1}$}] {};
\node at (1.5,0) [label=90:{$u^1_{i,kN}$}] {};

\node at (3.5,0) [label=90:{$u^2_{i,1}$}] {};
\node at (5,0) [label=90:{$u^2_{i,kN}$}] {};

\node at (8,0) [label=90:{$u^N_{i,1}$}] {};
\node at (9.5,0) [label=90:{$u^N_{i,kN}$}] {};

\node at (0.75,-0.8) {$b_i^1$};
\node at (4.25,-0.8) {$b_i^2$};
\node at (8.75,-0.8) {$b_i^N$};
\end{tikzpicture}

\vspace{0.7cm}
\begin{tikzpicture}
[node/.style={fill,circle,inner sep = 1.5pt},
dot/.style={fill,circle,inner sep = 0.5pt}]

\node at (-4,2) {b)};

\draw (-2,2) node[node,label=90:{$u^j_{i,1}$}] {};
\draw (-1.5,2) node[dot] {};
\draw (-1.3,2) node[dot] {};
\draw (-1.1,2) node[dot] {};
\draw (-0.6,2) node[node,label=90:{$u^j_{i,s}$}] (u) {};
\draw (0.2,2) node[dot] {};
\draw (0.4,2) node[dot] {};
\draw (0.6,2) node[dot] {};
\draw (1.4,2) node[node,label=90:{$u^j_{i,kN}$}] {};

\draw [color=gray] (-2.5,1.7) rectangle (1.9,2.9) {};
\node at (-1.2,1.2) {$b_i^j$};

\draw (4,2) node[node,label=90:{$u^{j'}_{i',1}$}] {};
\draw (4.8,2) node[dot] {};
\draw (5,2) node[dot] {};
\draw (5.2,2) node[dot] {};
\draw (6,2) node[node,label=90:{$u^{j'}_{i',s'}$}] (up) {};
\draw (6.5,2) node[dot] {};
\draw (6.7,2) node[dot] {};
\draw (6.9,2) node[dot] {};
\draw (7.4,2) node[node,label=90:{$u^{j'}_{i',kN}$}] {};

\draw [color=gray] (3.5,1.7) rectangle (7.9,2.9) {};
\node at (6.8,1.2) {$b_{i'}^{j'}$};

\draw (u) .. controls +(-90:0.8cm) and +(-90:0.8cm) .. (up);

\end{tikzpicture}
\caption{a) Associated with every color class $C_i$ is a block $B_i$ of vertex
sequences $b_i^1,\dots,b_i^N$ such that each $b_i^j$ is a sequence of $kN$
vertices $u_{i,1}^j,\dots,u_{i,kN}^j$. The vertex sequences $b_i^1,\dots,b_i^N$
correspond to the $N$ vertices of $C_i$. b) $(u_{i,s}^j, u_{i',s'}^{j'})$
is an edge in $E(G)$ if and only if $s = (i'-1)N+j'$, $s'=(i-1)N+j$,
and $(u_s, u_{s'}) \in E(M)$.}
\end{figure}

To construct $H = (V(H), E(H))$, we again associate with every color class
$C_i$ a block $B'_i$ of vertex sequences $b_i^{'1},\dots,b_i^{'N}$ such that each
$b_i^{'j}$ is a sequence of $kN$ vertices $v_{i,1}^{j},\dots,v_{i,kN}^{j}$.
$V(H)$ is the set of all vertices $v_{i,r}^{j}$ for $i = 1,\dots,k$,
$j=1,\dots,N$, and $r = 1,\ldots,kN$. The ordering $\prec_H$ of vertices
in $H$ is defined just as it was in $G$.
The edge set $E(H)$ is simply $\emptyset$: in this way, for any valid list embedding
$\phi : V(S) \to V(H)$, where $S$ is a subgraph of $G$, $V(S)$ must be an independent
set of $G$.

We complete the construction of the instance of p-\textsc{OLSE} by defining
the mapping $L: G \to 2^{V(H)}$. $L$ is a bijection that maps the vertices in each block
$B_i \in G$ to the block $B'_i \in H$ as follows:
$$
\forall u_{i,s}^j \in b_i^{j} \mbox{, } L(u_{i,s}^j) = \{v_{i,s}^{(N-j+1)}\}.
$$
This mapping is illustrated in Figure~3.

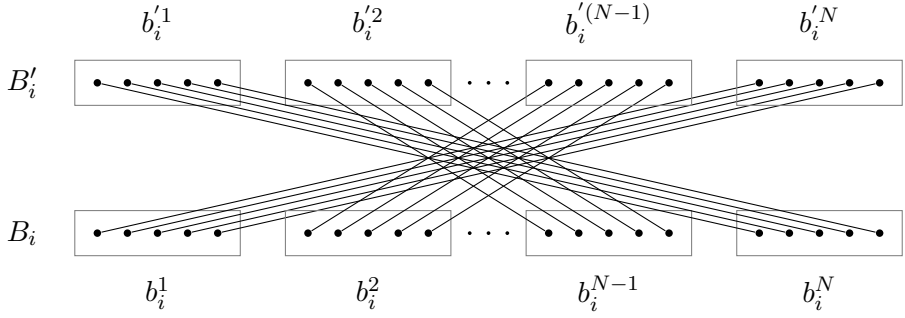
\begin{figure}
\label{fi:map}
\begin{tikzpicture}
[vertex/.style={fill,circle,inner sep = 1pt},
dot/.style={fill,circle,inner sep = 0.5pt}]

\node at (-1, 0) {$B_i$};
\node at (-1, 2) {$B_i'$};

\foreach \x in {0,2.8, 6,8.8}
{
  \draw [color=gray] (\x-0.3,-0.3) rectangle (\x+1.9,0.3);
  \draw [color=gray] (\x-0.3,1.7) rectangle (\x+1.9,2.3);
  \foreach \e in {0, 0.4, 0.8, 1.2, 1.6}
  {
    \node at (\x+\e, 0) [vertex] (h) {};
    \node at (8.8-\x+\e, 2) [vertex] (l) {};
    \draw (h) -- (l);
  }
}

\draw (4.95,2) node[dot] {};
\draw (5.2,2) node[dot] {};
\draw (5.45,2) node[dot] {};
\draw (4.95,0) node[dot] {};
\draw (5.2,0) node[dot] {};
\draw (5.45,0) node[dot] {};

\node at (0.8,-0.8) {$b_i^1$};
\node at (0.8,2.8) {$b_i^{'1}$};

\node at (3.6,-0.8) {$b_i^2$};
\node at (3.6,2.8) {$b_i^{'2}$};

\node at (6.8,-0.8) {$b_i^{N-1}$};
\node at (6.8,2.8) {$b_i^{'(N-1)}$};

\node at (9.6,-0.8) {$b_i^N$};
\node at (9.6,2.8) {$b_i^{'N}$};
\end{tikzpicture}
\caption{The mapping $L$ maps vertices of $B_i$ to vertices of $B_i'$. More
precisely, it maps vertices of $b_i^j$ to vertices of $b_i^{'N-j+1}$ in a way
that forbids an embedding that would map a vertex in $b_i^j$ and a vertex in
$b_i^{j'}$ simultaneously, for any $i=1,\dots,k$, $j,j'=1,\dots,N$, $j \not= j'$.}
\end{figure}

This completes the construction of the instance of p-OLSE. Observe that in the constructed instance we have $\Delta_H=1$, $\Delta_G=1$, and $\Delta_L=1$. We show next that $M$ has an independent set of size $k$ if and only if the constructed instance of p-OLSE has a subgraph $S$ of $G$ of size $k^2N$, and a valid embedding
$\phi: V(S) \to V(H)$.

Let $I$ be solution to the instance $(M, f)$ of \textsc{$k$-MCIS}. Each vertex $u_s \in I$ corresponds to a
sequence $b_i^j \in B_i$ where $s = (i-1)N+j$.
For any two sequences $b_i^j \in B_i, b_{i'}^{j'} \in B_{i'}$, there is an edge
between a vertex of $b_i^j$ and a vertex of $b_{i'}^{j'}$ if and only if
$(u_s, u_{s'}) \in E(M)$ for $s' = (i-1)N+j$ and $s = (i'-1)N+j'$ ($i < i'$). Since for
every pair $u_s, u_{s'} \in I, (u_s, u_{s'}) \notin E(M)$ and since
$f(u_s) \ne f(u_{s'})$, there are $k$ sequences $b_i^j$ with the property
that for any two vertices $u_{i,s}^j \in b_i^j$ and
$u_{i',s'}^{j'} \in b_{i'}^{j'}$, $(u_{i,s}^j,u_{i',s'}^{j'}) \notin E(G)$.
By construction of $G$, the vertices of any two sequences $b_i^j, b_{i'}^{j'}$
($i \ne i'$) that do not have edges between them are simultaneously embeddable in $H$. Therefore,
we can define a valid embedding $\phi : V(S) \to V(H)$, where $|S| = k \cdot (kN) =k^2N$, and $S$ consists of all the vertices in each $b_i^j$
that corresponds to a vertex in $I$.

Conversely, let $\phi: V(S) \to V(H)$, $|S| = k^2N$, be a valid embedding
from a subgraph $S$ of $G$ into $H$.  By construction of $L$, $\phi$ can embed at most $kN$
vertices from any block $B_i$ into $H$ and, furthermore, all vertices
in $B_i$ embedded by $\phi$ must belong to the same sequence $b_i^j$ of
$B_i$. Therefore, exactly $k$ sequences, one from each of the $k$ blocks, are fully embedded by $\phi$. Moreover, because $\Delta_H = 0$, $\phi$ embeds all $kN$ vertices of
the $k$ sequences $b_i^j$ in $S$ if and only if for any two vertices
$u_{i,s}^j, u_{i',s'}^{j'} \in S, (u_{i,s}^j, u_{i',s'}^{j'}) \notin E(G)$.
By construction of $G$, each sequence $b_i^j$ corresponds to the $j^{th}$
vertex of the color class $C_i$. Therefore, we can find a subset
$I \subseteq M, |I| = k$ such that for any two vertices $u_s, u_{s'} \in I,
f(u_s) \ne f(u_{s'})$ and $(u_s, u_{s'}) \notin E(M)$; that is, the set
$I$ is a solution to the instance $(M, f)$ of \textsc{$k$-MCIS}.

 For p-OLISE, since $\Delta_H =0$, the same reduction given above can be used to show its $\NP$-hardness when $\Delta_H=0$, $\Delta_G=1$ and $\Delta_L=1$.  To show that p-OLISE is \NP-hard when $\Delta_H=1$, $\Delta_G=0$ and $\Delta_L=1$ we need only tweak the construction of the p-OLISE instance as follows.  We construct $H$ as we constructed $G$ in the previous case (with edges), and we construct $G$ as we constructed $H$ in the previous case (without edges). The \NP-hardness when $\Delta_H=0$, $\Delta_G=1$ and $\Delta_L=1$ follows by symmetry.
\end{proof}

We contrast this hardness result with a result for a version of the problem where the injective map need not preserve the ordering of the vertices.  First, we define the parameterized {\sc List Subgraph Embedding} problem (p-LSE):

\paramproblem{} {Two graphs $G$ and $H$; a function $L : V(G) \longrightarrow 2^{V(H)}$; and $k \in \nat$}{$k$}{Is there a subgraph $S$ of $G$ of $k$ vertices and an injective map $\phi: V(S) \longrightarrow V(H)$ such that: (1) $\phi(u) \in L(u)$ for every $u \in S$ and (2) for every $u, u' \in S$, if $uu' \in E(G)$ then $\phi(u)\phi(u') \in E(H)$} \\

\begin{proposition}\label{prop:pLSE-in-P}
The p-LSE problem restricted to instances where $\Delta_G=1$, $\Delta_H=0$, and $\Delta_L=1$ is in $\Pol$.
\end{proposition}
\begin{proof}
It suffices to prove that the optimization version opt-LSE of p-LSE is in $\Pol$.

For each vertex $v \in H$ define $L^{-1}(v) = \{u \in G : L(u) = v\}$. Note that for each $v \in H$, $L^{-1}(v)$ defines a subset of vertices in $V(G)$ from which only one vertex can be included in the solution.  Let $\mathcal{S} = \{S_i : i = 1,\cdots,|V(H)|\}$ be the set of subsets of $V(G)$ where $S_i = L^{-1}(v_i), v_i \in V(H)$. Since $\Delta_L=1$, the elements of $\mathcal{S}$ are pairwise disjoint. We say that two subsets $S_i$ and $S_j$, $i \neq j$, are {\em adjacent} if there exists an edge $uu' \in E(G)$ where $u \in S_i$ and $u' \in S_j$.  We define a {\em cycle} of subsets to be a sequence of distinct subsets $C = \langle S_1, \ldots, S_l \rangle$, $\l > 2$, such that $S_i$ and $S_{i+1}$ are adjacent, for $i=1, \ldots, l-1$, and $S_l$ and $S_1$ are adjacent; or a sequence of two subsets $C = \langle S_i, S_j \rangle$ such that there are two distinct vertices $u_i, u'_i$ in $S_i$ and two distinct vertices $u_j, u'_j$ in $S_j$ such that $u_i$ is adjacent to $u_j$ and $u'_i$ is adjacent to $u'_j$.  Using these constructs, we can define the following rules to simplify the problem:

\begin{enumerate}[(i)]
  \item For each subset $S_i \in \mathcal{S}$ with a vertex $u$ such that $|N_{G}(u)| = 0$, or $u$ has a neighbor in $S_i$, we add $u$ to the solution $S$ and we define $\phi(u) = v_i$. We remove the vertices in $S_i$ from $G$ and vertex $v_i$ from $H$.

  \item For each cycle of subsets $C = \langle S_1, \ldots, S_l \rangle$, if $l > 2$, choose two vertices $u_i, u_i'$ in each $S_i$, $i=1, \ldots, l$ such that $u_i'$ is adjacent to $u_{i+1}$, for $i=1, \ldots, l-1$, and $u'_l$ is adjacent to $u_1$. Add $u_i$ to the solution $S$, define $\phi(u_i)=v_i$, and remove the vertices in $S_i$ from $G$ and $v_i$ from $H$, for $i=1, \ldots, l$. If $l=2$, let $u_1, u'_1 \in S_1$ and $u_2, u'_2 \in S_2$ such that $u_1$ is adjacent to $u_2$ and $u'_1$ is adjacent to $u'_2$. Add $u_1, u'_2$ to the solution $S$, define $\phi(u_1)=v_1$, $\phi(u'_2)=v_2$, and remove the vertices in $S_1 \cup S_2$ from $G$ and $v_1, v_2$ from $H$.
\end{enumerate}

Clearly, each of the above rules is safe in the sense that it is possible to extend the solution obtained after each application of the rule to an optimal solution. After the above rules are no longer applicable, we are left with a graph $G$ (we call the resulting graph $G$ for simplicity) where the remaining subsets in $\mathcal{S}$ (considered as nodes), together with their adjacency relation, form a tree-like structure.  We say that a subset $S_i \in \mathcal{S}$ is a {\em leaf-subset} (in $\mathcal{S}$) if it is adjacent to exactly one other subset in $\mathcal{S}$; observe that every subset $S_i$ must be adjacent to at least one other subset in $\mathcal{S}$ by rule (i). Since each leaf-subset is adjacent to one subset in $\mathcal{S}$, each leaf-subset must consist of exactly one vertex for the following reasons. First, by rule (i), each vertex in a leaf-subset $S_i$ must have a neighbor in the (single) subset $S_j$ that is adjacent to $S_i$. By rule (ii), $S_i$ must consist of exactly one vertex; otherwise, rule (ii) would apply to $S_i$ and $S_j$. To complete the construction of our solution, we iteratively apply a new rule, rule (iii), that includes the vertex $u_i \in S_i$ in our solution $S$, for each leaf-subset $S_i$, defines $\phi(u_i) = v_i$, and removes the vertices in $S_i$ from $G$ and $v_i$ from $H$. Rule (iii) is safe because any optimal solution must contain exactly one of $u_i$ and its only neighbor $w$ in the subset adjacent to $S_i$, and in the case when the solution contains $w$, $w$ can be exchanged with $u_i$, and the map $\phi$ can be updated by defining $\phi(u_i)=v_i$. When rule (iii) is no longer applicable, the graph $G$ is empty, and the algorithm has computed an optimal solution.  Clearly, the running time is polynomial because each application of a rule removes at least one vertex from $G$ (and from $H$). \end{proof}

\section{Approximation Results}\label{subsec:approximation}
In this section we consider opt-OLSE and opt-OLISE with respect to the framework of approximation theory. We begin with the following proposition:

\begin{proposition}\label{prop:apx}
The opt-OLSE problem restricted to instances in which $\Delta_G=O(1)$ has an approximation algorithm of ratio $(\Delta_G+1)$, and the opt-OLISE problem restricted to instances in which $\Delta_G=O(1)$ and $\Delta_H=O(1)$ has an approximation algorithm of ratio $(\Delta_H+1) \cdot (\Delta_G+1)$.
\end{proposition}

\begin{proof}
Let $(G, H, \prec_G, \prec_H, L)$ be an instance of opt-OLSE, and consider the following algorithm. Apply the dynamic programming algorithm in Proposition~\ref{prop:noedges} to $(G, H, \prec_G, \prec_H, L)$ after removing the edges of $G$ and the edges of $H$, and let $S$ and $\phi$ be the subgraph and map obtained, respectively. Apply the following trivial approximation algorithm to compute an independent set $I$ of $S$: pick a vertex $u$ in $S$, include $u$ in $I$, remove $u$ and $N(u)$ from $S$, and repeat until $S$ is empty.  Return the subgraph $G[I]=I$, and the restriction of $\phi$ to $I$, $\phi_I$. Clearly, we have $|I| \geq |S|/(\Delta_G + 1)$.

Since $\phi$ is a valid list embedding of $S$ after the edges of $G$ and $H$ have been removed, and since $I$ is an independent set of $G[S]$, it is clear that $\phi_I$ is a valid list embedding of $G[I]$ into $H$. Therefore, the algorithm is an approximation algorithm. Now let $S_{opt}$ be an optimal solution of the instance. $|S|$ is an upper bound on the size of an optimal solution $S_{opt}$ since $S_{opt}$ is a solution to opt-OLSE that respects the embedding constraint. Therefore, $|S_{opt}|/|I| \leq |S|/|I| \leq (\Delta_G + 1)$.

For opt-OLISE, we apply the algorithm above to obtain $I$, however, there may be vertices $u,v \in I$ where $\phi_I(u)\phi_I(v) \in E(H)$, violating the embedding constraint of opt-OLISE.  We use a similar method to that used to obtain the approximation algorithm described above to overcome this hurdle. This time we pick a vertex $v$ in the image $\phi(I)$ of the independent set $I$, include $v$ in the independent set $I' \subseteq H$, remove $v$ and $N(v)$ from $\phi(I)$, and repeat until $\phi(I)$ is empty. As was the case for $I$, it follows that $|I'| \geq |\phi(I)|/(\Delta_H + 1)$. Let $I'' \subseteq G$ be the inverse image $\phi_I^{-1}(I')$. Since $|I| \geq |S|/(\Delta_G + 1)$ and $|I'| \geq |\phi(I)|/(\Delta_H + 1)$, it follows that $|I''| \geq |S|/((\Delta_G + 1)(\Delta_H + 1))$. $|S|$ remains an upper bound on the size of an optimal solution $S_{opt}$. Therefore, $|S_{opt}|/|I''| \leq |S|/|I''| \leq (\Delta_G + 1)(\Delta_H + 1)$. \end{proof}

The inapproximability results outlined in Table~1 follow by simple reductions from the {\sc Maximum Independent Set} problem.  We introduce the following reduction modeled after a reduction in \cite{xiuzhen} for a variation of opt-OLSE.

\begin{lemma}
\label{lem:reduce-IS}
Let $G$ be a graph.  Construct an instance $I_G = (G, H, \prec_G, \prec_H, L)$ of opt-OLSE from $G$ as follows. Let $\prec_G$ be an arbitrary ordering on the vertices of $G$, ${u_1, ..., u_n}$ where $u_1 \prec_G u_2 \dots \prec_G u_n$.  Let $H$ be the graph consisting of the isolated vertices ${v_1, ..., v_n}$, where $v_1 \prec_H v_2 \dots \prec_H v_n$.  Define the mapping $L : V(G) \to 2^{V(H)}$ such that for each vertex $u_i \in V(G)$, $L(u_i) = \{v_i\}$. We have:
\begin{enumerate}[(i)]
  \item $G$ has an independent set of size $k$ if and only if $I_G$ has a solution of size $k$.
  \item There is an fpt-reduction from \textsc{Independent Set} to p-OLSE where $\Delta_G=\infty, \Delta_H=0$ and $\Delta_L=1$, and to p-OLISE where $\Delta_G=\infty$, $\Delta_H=0$ (resp. $\Delta_G=0$, $\Delta_H=\infty$ by symmetry) and $\Delta_L=1$.
  \item There is an L-reduction from \textsc{Maximum Independent Set} to opt-OLSE where $\Delta_G=\infty, \Delta_H=0$ and $\Delta_L=1$, and to opt-OLISE where $\Delta_G=\infty$, $\Delta_H=0$ (resp. $\Delta_G=0$, $\Delta_H=\infty$ by symmetry) and $\Delta_L=1$.
\end{enumerate}
\end{lemma}

\begin{proof}

(i) Let $I \subseteq V(G)$ be an independent set in $G$ where $|I| = k$.  Since $\Delta_H = 0$, for any subgraph $S$ of $G$ and a valid embedding $\phi : V(S) \to V(H)$, no two vertices $u, u' \in S$ are adjacent. Given that $I$ is an independent set and that for each vertex $u_i \in V(G), L(u_i) = \{v_i\}$, letting $S = I$, we can define the embedding $\phi : V(S) \to V(H)$, where $\phi(u_i) = v_i$ for $u_i \in S$, satisfying that for any two vertices $u_i, u_j \in S$ if $u_i \prec_G u_j$ then $\phi(u_i) \prec_H \phi(u_j)$.  Thus, there is a valid embedding from $S$ to $V(H)$ where $|S| = |I| = k$.

Conversely, let $\phi : V(S) \to V(H)$ be an embedding where $S$ is a subgraph of $G$ satisfying $|S| = k$.  As observed earlier, $S$ must be an independent set since $\Delta_H = 0$. Thus, $G$ has an independent set $I$ where $|I| = |S| = k$.

(ii) Let $(G,k)$ be an instance of \textsc{Independent Set}. We construct the instance \newline $I_G=(G, H, \prec_G, \prec_H, L, k)$ by following the same procedure outlined in the statement of the lemma.  By part (i), $G$ contains an independent set $I$ where $|I| = k$ if and only if there exists a valid embedding $\phi : V(S) \to V(H)$ where $S = I$. Since the same parameter is shared between the instance of \textsc{Independent Set} and the instance $I_G$ of p-OLSE, and since the construction of $I_G$ can be clearly done in polynomial time, we have an fpt-reduction from \textsc{Independent Set} to p-OLSE. The reduction extends immediately to the case of p-OLISE where $\Delta_G = \infty$, $\Delta_H = 0$  and $\Delta_L = 1$ because $\Delta_H = 0$.  By symmetry, it also can be used for the case of p-OLISE where $\Delta_G = 0$, $\Delta_H = \infty$ and $\Delta_L = 1$.

(iii) The $L$-reduction is given by the function $h$ that maps the instance $G$ of \textsc{Maximum Independent Set} into the instance $I_G$ of opt-OLSE, as described in the statement of the lemma, and the function $h'$ that, for each solution $S, \phi$ of $I_G$ corresponds the solution $I=V(S)$ of $G$, and the constants $\alpha = \beta =1$. By part (i), any solution $I$ of $G$ corresponds to a solution $S, \phi$ of $I_G$ such that $|I| = |S|$, and vice versa. Therefore, $opt(G) = opt(I_G)$, and $|I| - opt(G) = |S| - opt(I_G)$. This shows that \textsc{Maximum Independent Set} $L$-reduces to opt-OLSE.

The $L$-reduction extends immediately to the case of opt-OLISE where $\Delta_G=\infty, \Delta_H=0$ and $\Delta_L=1$, and by symmetry it extends to the case of opt-OLISE where $\Delta_G=0, \Delta_H=\infty$ and $\Delta_L=1$. \end{proof}

We can now demonstrate that opt-OLSE is APX-complete when $\Delta_G=O(1)$ and opt-OLISE is APX-complete when $\Delta_G=O(1)$ and $\Delta_H=O(1)$.

\begin{proposition}\label{prop:apx-complete}
The opt-OLSE problem restricted to instances in which $\Delta_G=O(1)$ and the opt-OLISE problem restricted to instances in which $\Delta_G=O(1)$ and $\Delta_H=O(1)$ are APX-complete.
\end{proposition}

\begin{proof}
Proposition~\ref{prop:apx} showed that opt-OLSE is in APX when $\Delta_G=O(1)$. We demonstrate that opt-OLSE is APX-complete when $\Delta_G=O(1), \Delta_H=0$ and $\Delta_L=1$ to yield the result. Since the case of opt-OLSE where $\Delta_G=O(1)$, $\Delta_H=0$ and $\Delta_L=1$ is a subset of opt-OLSE where $\Delta_G=\infty$, $\Delta_H=0$ and $\Delta_L=1$ we can apply Lemma~\ref{lem:reduce-IS} part (iii) to $L$-reduce {\sc Maximum Independent Set} to the cases of opt-OLSE where $G$ has bounded degree.  Since {\sc Maximum Independent Set} on bounded degree graphs is APX-complete, it follows that opt-OLSE is APX-complete when $\Delta_G=O(1)$~\cite{py91}.  Similar arguments can be applied to opt-OLISE to show that it is APX-complete when restricted to instances in which $\Delta_G=O(1)$ and $\Delta_H=O(1)$.
\end{proof}

The inapproximability results for opt-OLSE and opt-OLISE when $\Delta_G = \infty$, $\Delta_L=1$ and $\Delta_H = 0$ (and the symmetric case for opt-OLISE when  $\Delta_G = 0$, $\Delta_L=1$ and $\Delta_H =\infty$) also follow from Lemma~\ref{lem:reduce-IS}.

\begin{proposition}\label{prop:no-apx}
The opt-OLSE problem restricted to instances in which $\Delta_G=\infty, \Delta_H=0$ and $\Delta_L=1$ and opt-OLISE restricted to instances in which $\Delta_G=\infty, \Delta_H=0$ and $\Delta_L=1$ or $\Delta_G=0, \Delta_H=\infty$ and $\Delta_L=1$ cannot be approximated to within a factor of $n^{\frac{1}{2}-\epsilon}$ unless $\Pol = \NP$, where $n$ is the number of vertices in $G$.
\end{proposition}

\begin{proof}
H\r{a}stad demonstrated in~\cite{hastad97} that {\sc Maximum Independent Set} is not approximable to within a factor of $n^{\frac{1}{2}-\epsilon}$, for any $\epsilon > 0$, unless $\Pol = \NP$.  The proposition follows immediately from the $L$-reduction in part (iii) of Lemma~\ref{lem:reduce-IS} for the given cases of opt-OLSE and opt-OLISE.
\end{proof}

\section{Parameterized Complexity Results}\label{subsec:parameterized}

In this section we consider p-OLSE and p-OLISE with respect to the framework of parameterized complexity. Applying part (ii) of Lemma~\ref{lem:reduce-IS} to p-OLSE and p-OLISE yields the following result:

\begin{proposition}
\label{prop:whardness1}
The p-OLSE problem restricted to instances in which $\Delta_G = \infty$, $\Delta_H = 0$ and $\Delta_L = 1$ is $W[1]$-complete, and the p-OLISE problem restricted to instances in which $\Delta_G = \infty$, $\Delta_H = 0$ (resp. $\Delta_G = 0$ and $\Delta_H =\infty$ by symmetry) and $\Delta_L = 1$ is $W[1]$-complete.
\end{proposition}

\begin{proof}
Downey and Fellows showed in~\cite{df2} that \textsc{Independent Set} is $W[1]$-complete.  By part (ii) of Lemma~\ref{lem:reduce-IS} there is an fpt-reduction from \textsc{Independent Set} to p-OLSE restricted to instances in which $\Delta_G = \infty$, $\Delta_H = 0$ and $\Delta_L = 1$, and to p-OLISE restricted to instances in which $\Delta_G = \infty$, $\Delta_H = 0$ (resp. $\Delta_H =\infty$ and $\Delta_G = 0$ by symmetry) and $\Delta_L = 1$. The statement of the proposition follows.
\end{proof}

Next, we consider the case of p-OLSE and p-OLISE when $\Delta_G=1$, $\Delta_H = 1$ and $\Delta_L = \infty$. Evans~\cite{evans} reduces the $W[1]$-complete problem {\sc Clique}~\cite{fptbook} to the {\sc Longest Arc-Preserving Common Subsequence} (LAPCS) to show that LAPCS is $W[1]$-hard when the length of the common subsequence, $l$, is the parameter, and when no two arcs in the instance share an endpoint. The reduction encodes a graph $H$ as an arc-annotated sequence $(S_1, P_1)$ where each vertex in $V(H)$ is a substring of $S_1$ of the form $(ba^nb)$ and where the edges in $H$ are encoded as arcs between the substrings of $S_1$ with endpoints corresponding to the endpoints of the edge in $H$.  A clique of size $k$ is encoded as an arc-annotated sequence $(S_2, P_2)$ where each of the $k$ vertices in the clique is encoded as a substring of $S_2$ of the form $(ba^kb)$ with similarly constructed arcs between substrings.  A clique of size $k$ exists in $H$ if $(S_2, P_2)$ is an arc-preserving subsequence of $(S_1, P_1)$ with length $l = k(k+2)$.  Since the length of $S_2, l$, is a function of $k$, {\sc Clique} is fpt-reducible to LAPCS.

We show that LAPCS is a special case of p-OLISE where $\Delta_G=1, \Delta_H=1$ and $\Delta_L=\infty$.  Taking each character in $S_1$ and $S_2$ to be a vertex in $V(H)$ and $V(G)$, respectively, and preserving the arcs in $P_1$ and $P_2$ as edges in $E(H)$ and $E(G)$, respectively, $(S_1, P_1)$ and $(S_2, P_2)$ are reduced to the graphs $H$ and $G$, respectively, and $|V(G)| = k(k+2)$. Each vertex $u \in V(G)$ has a list in $V(H)$ that consists of the vertices $v \in V(H)$ such that the character for $u \in S_2$ is the same as the character for $v \in S_1$.  The size of the parameter $l$ remains the same, completing the reduction of LAPCS to p-OLISE.  Note that a clique of size $k$ is encoded by the property that the embedding from $G$ to $H$ preserves the arcs in $G$, which is a property of solutions for both p-OLSE and p-OLISE.  Thus, by the same arguments, LAPCS reduces to p-OLSE as well.  Finally, note that this reduction results in instances of p-OLSE and p-OLISE where the question is whether we can embed the whole graph $G$ into $H$, yielding the following proposition:

\begin{proposition}\label{prop:whardness}
The p-OLSE and p-OLISE problems restricted to instances in which $\Delta_H = 1$ and $\Delta_G = 1$ are $W[1]$-complete.  Moreover, the problems remain $W[1]$-complete when the parameter $k$ is the number of vertices in $V(G)$.
\end{proposition}

\begin{proposition}\label{prop:whardness2}
The p-OLISE problem restricted to instances where $\Delta_H = 1$, $\Delta_G = 1$, and $\Delta_L = 1$ is $W[1]$-complete under Turing fpt-reductions.
\end{proposition}

\begin{proof}
Let ${\cal I}$ be the set of instances of p-OLISE where $\Delta_H = 1$, $\Delta_G = 1$, and $\Delta_L = 1$.  For an instance $I = (G, H, \prec_G, \prec_H, L, k)\in {\cal I}$ define $I^{-1} = (H, G, \prec_H, \prec_G, L^{-1}, k)$ where, for each $v \in V(H), L^{-1}(v) = \{u : L(u) = \{v\}\}$. As a consequence of $\Delta_L = 1$, for every instance $I^{-1}$, and for every vertex $u \in G$, we have $|\{v \in V(H): L^{-1}(v) = \{u\}\}| = 1$. Let ${\cal I}^{-1}$ be the set of all $I^{-1}$.  It is easy to see that an instance of p-OLISE $I$ is a yes-instance if and only if $I^{-1}$ is a yes-instance.  Therefore, to prove the proposition we can equivalently prove that p-OLISE restricted to instances in ${\cal I}^{-1}$ is $W[1]$-complete.  Denote by p-OLISE-simple p-OLISE restricted to instances in which $\Delta_H = 1$, $\Delta_G = 1$, $\Delta_L = \infty$, and $|V(G)| = k$.  Proposition~\ref{prop:whardness} showed that p-OLISE-simple is $W[1]$-complete.  We give a Turing fpt-reduction from p-OLISE-simple to ${\cal I}^{-1}$ to yield the result.

Let $I' = (G, H, \prec_G, \prec_H, L)$ be an instance of p-OLISE-simple where $V(G) = {u_1, ..., u_k}$.   Using color-coding, there is a family ${\cal F}$ of fpt-many $k$-colorings for the vertices in $V(H)$ using the colors ${u_1, \ldots , u_k}$ --- corresponding to the labels of the vertices of $G$ --- such that if $I'$ is a yes-instance, then there is a coloring $c \in {\cal F}$, and a valid embedding  $\phi_c : V(G) \to V(H)$, such that $c(v_i) = u_i$ if and only if $\phi_c(u_i) = v_i$.  To reduce $I'$, for every coloring $c \in {\cal F}$ we create the instance $I^c$ of p-OLISE where $I^c = (G, H, \prec_G, \prec_H, L^{c}, k)$, where for each $u \in V(G)$, $L^{c}(u) = \{v \in H: c(v) = u\}$.  Observe that by the definition of $L^c$, for each vertex $v \in V(H)$, $|\{u : L^c(u) = \{v\}\}| = 1$.  Thus, $I^c \in {\cal I}^{-1}$.

Hence there is a Turing fpt-reduction from p-OLISE-simple to ${\cal I}^{-1}$, and if there is an algorithm that runs in fpt-time that solves ${\cal I}^{-1}$, then we can find a solution for p-OLISE-simple in fpt-time by enumerating the colorings in ${\cal F}$ and applying the algorithm for ${\cal I}^{-1}$. Given that p-OLISE-simple is $W[1]$-complete, it follows that ${\cal I}^{-1}$ is $W[1]$-complete under Turing fpt-reductions.  It follows that ${\cal I}$ is $W[1]$-complete as well under Turing fpt-reductions.
\end{proof}

Again, relaxing the ordering constraint significantly reduces the difficulty of the problem.  For example, when we relax the ordering constraint in p-OLISE we have the parameterized {\sc List Induced Subgraph Embedding} problem, shortly (p-LISE):

\paramproblem{} {Two graphs $G$ and $H$; a function $L : V(G) \longrightarrow 2^{V(H)}$; and $k \in \nat$}{$k$}{Is there a subgraph $S$ of $G$ of $k$ vertices and an injective map $\phi: V(S) \longrightarrow V(H)$ such that: (1) $\phi(u) \in L(u)$ for every $u \in S$ and (2) for every $u, u' \in S$, $uu' \in E(G)$ if and only if $\phi(u)\phi(u') \in E(H)$} \\

\begin{proposition}\label{prop:unordered_is_in_P}
The p-LISE problem restricted to instances in which $\Delta_G=1$ and $\Delta_H=1$ is in $\Pol$.
\end{proposition}

\begin{proof}
It suffices to show that the optimization (maximization) version opt-LISE of p-LISE is in $\Pol$. We proceed by reducing opt-LISE in the case where $\Delta_G=1$ and $\Delta_H=1$ to the well-known polynomial-time computable problem \textsc{Weighted Maximum Matching on Bipartite Graphs} (WMMBG):  Given a graph $B(V, E)$, a bipartition $V(B) = (X, Y)$, and weight function $w : E(B) \to \mathbb{R}$, find a matching $M \subseteq E(B)$ with maximum weight, where the weight of a matching $w(M) = \sum_{e \in M} w(e)$.

Let $(G, H, L, k)$ be an instance of opt-LISE.  To create $B$, we first create the vertices of $X$ from the edges and vertices of $G$.  For each edge $e \in E(G)$ where $e = uu'$, we add a vertex to $X$.  For each vertex $u \in G$ where $|N_G(u)| = 0$, we add a vertex to $X$.  We follow the same procedure to create the vertices of $Y$ from $H$.  To complete $B$ we add edges as follows:

\begin{enumerate}[(i)]
 \item For any two edges $e_1=uu' \in E(G)$ and $e_2=vv' \in E(H)$ such that $v \in L(u)$ and $v' \in L(u')$, we add an edge of weight 2 between the vertex in $X$ corresponding to $e_1$ and the vertex in $Y$ corresponding $e_2$.
 \item For each vertex $u \in G$ such that $|N_G(u)| = 0$ and for each vertex $v \in L(u)$ such that $|N_H(v)| = 0$, we add an edge of weight 1 between the vertex corresponding to $u$ in $X$ and the vertex corresponding to $v$ in $Y$.

 \item For each vertex $u \in G$ such that $|N_G(u)| = 0$, and for each vertex $v \in L(u)$ with neighbor $v' \in H$, we add an edge of weight 1 between the vertex corresponding to $u$ in $X$ and the vertex corresponding to the edge $vv'$ in $Y$. If both $v$ and $v'$ appear in $L(u)$ we add {\em only} one edge between the vertex corresponding to $u$ in $X$ and that corresponding to $vv'$ in $Y$ so that no multi-edges are created in $B$.

 \item Finally, for each vertex $u \in G$ with neighbor $u'$ in $G$, and for each vertex $v \in L(u)$ where $|N_H(v)| = 0$, we add an edge of weight 1 between the vertex in $X$ corresponding to $uu'$ and the vertex in $Y$ corresponding to $v$. If $v$ appears in both $L(u)$ and $L(u')$, we add {\em only} one edge between the vertex corresponding to $uu'$ in $X$ and that corresponding to $v$ in $Y$ so that no multi-edges are created in $B$.
\end{enumerate}

This completes the construction of $B$.  We show next that $B$ has a matching $M$ where $w(M) = |S|$ if and only if there is a subset $S \subseteq V(G)$ where $|S| = w(M)$ and there is a valid embedding $\phi : V(S) \to V(H)$.  We begin by showing that for every matching $M \subseteq E(B)$ we can construct a subgraph $S$ of $G$ with a valid embedding $\phi: V(S) \to V(H)$ such $|S| = w(M)$.  Given a matching $M$, for each edge $e \in M$ we add vertices to $S$ and define $\phi(u)$ for each vertex $u$ added using the following rules:

\begin{enumerate}[(i)]
  \item Let $e=xy \in M$ be an edge of weight 2, where $x \in X$ corresponds to an edge $uu' \in G$ and $y\in Y$ corresponds to an edge $vv' \in H$. By the construction of $B$, one of the two vertices $v, v'$, say $v$, must be in $L(u)$ and the other $v' \in L(u')$.  Add $u$ and $u'$ to $S$, define $\phi(u) = v$, and define $\phi(u') = v'$.
  \item Let $e=xy \in M$ be an edge of weight 1 where $x \in X$ corresponds to a vertex $u \in V(G)$, and $y \in Y$ corresponds to a vertex $v \in V(H)$.  By the construction of $B$, $v \in L(u)$, and $|N_G(u)| = |N_H(v)| =0$.  Add $u$ to $S$ and define $\phi(u) = v$.
  \item Let $e=xy \in M$ be an edge of weight 1, where $x \in X$ corresponds to a vertex $u \in V(G)$ and $y \in Y$ corresponds to an edge $vv' \in H$. By the construction of $B$, $v \in L(u)$ (without loss of generality).  Add $u$ to $S$ and define $\phi(u) = v$.
  \item Let $e=xy \in M$ be an edge of weight 1, where $x \in X$ corresponds to an edge $uu' \in G$, and $y \in Y$ corresponds to a vertex $v \in V(H)$.  By the construction of $B$, $v \in L(u)$ (without loss of generality).  Add $u$ to $S$ and define $\phi(u) = v$.
\end{enumerate}

The resulting subgraph $S$ satisfies $|S| = w(M)$ because for each edge $e \in M$ where $w(e) = 2$, we add two vertices to $S$, and for each edge $e \in M$ where $w(e) = 1$ we add one vertex to $S$. Moreover, the resulting map $\phi: V(S) \to V(H)$ is a valid embedding. In particular, according to the above rules, for each pair of vertices $u,u' \in V(G)$ added to $S$, $uu' \in E(G)$ if and only if $\phi(u)\phi(u') \in E(H)$.

It remains to show that for a subgraph $S$ of $G$ with a valid embedding $\phi : V(S) \to V(H)$ there is a matching $M \in E(B)$ such that $w(M) = |S|$.  Similarly to how a subgraph $S$ of $G$ and embedding $\phi : V(S) \to V(H)$ were defined from a matching $M \subseteq E(B)$, there are a few rules by which we can construct a matching $M \subseteq E(B)$ from a subgraph $S$ of $G$ and a valid embedding $\phi : V(S) \to V(H)$:

\begin{enumerate}[(i)]
  \item For an edge $e= uu' \in S$, $\phi(u)\phi(u')$ must be an edge in $H$. By the construction of $B$, there is a corresponding vertex $x \in X$ for $uu'$, a corresponding vertex $y \in Y$ for the edge $\phi(u)\phi(u') \in H$, and an edge $e=xy \in E(B)$ of weight 2.  Add $e$ to $M$.

  \item For a vertex $u \in S$ where $|N_G(u)| = 0$ and $|N_H(\phi(u))| = 0$, by the construction of $B$, there is a corresponding vertex $x \in X$ for $u$, a corresponding vertex $y \in Y$ for $\phi(u)$, and an edge $e=xy \in E(B)$ of weight 1.  Add $e$ to $M$.

  \item For an edge $uu' \in G$ where $u \in S$, $u' \notin S$, and $|N_H(\phi(u))| = 0$, by the construction of $B$, there is a corresponding vertex $x \in X$ for $uu'$, a corresponding vertex $y \in Y$ for $\phi(u)$, and an edge $e=xy \in E(B)$ of weight 1.  Add $e$ to $M$.

  \item For a vertex $u \in S$ where $|N_G(u)| = 0$ and an edge $vv' \in E(H)$ where $v = \phi(u)$, by the construction of $B$, there is a vertex $x \in X$ that corresponds to $u$, a vertex $y \in Y$ that corresponds to $vv'$, and an edge $e=xy \in E(B)$ of weight 1.  Add $e$ to $M$.
\end{enumerate}

Since $\phi$ is a valid embedding that embeds $S$ into $H$, it is not difficult to verify that the set of constructed edges $M$ is a matching, and that $w(M) = |S|$. This completes the proof. \end{proof}

We turn our attention next to discussing the fixed-parameter tractability results for p-OLSE when both $\Delta_G$ and $\Delta_L$ are $O(1)$ ($\Delta_H$ may be unbounded). Let $(G, H, \prec_G, \prec_H, L, k)$ be an instance of p-OLSE in which both $\Delta_G$ and $\Delta_L$ are upper bounded by a fixed constant. Consider the graph ${\cal G}$ whose vertex-set is $V(G) \cup V(H)$ and whose edge-set is $E(G) \cup E(H) \cup E_L$, where $E_L=\{uv \mid u \in G, v \in H, v \in L(u)\}$; that is, ${\cal G}$ is the union of $G$ and $H$ plus the edges that represent the mapping $L$. We perform the following {\em splitting} operation on the vertices of ${\cal G}$ (see Figure~\ref{fig:splitting} for illustration):

\begin{definition}\label{def:splitting}
Let $u$ be a vertex in ${\cal G}$ and assume that $u \in G$ (the operation is similar when $u \in H$). Suppose that the vertices of $G$ are ordered as $\langle u_1, \ldots, u_n\rangle$ with respect to $\prec_G$, and suppose that $u=u_i$, for some $i \in \{1, \ldots, n\}$. Let $e_1=uv_{i_1}, \ldots, e_r=uv_{i_r}$ be the edges incident to $u$ in $E_L$, and assume that $v_{i_1} \prec_H v_{i_2} \prec_H \ldots \prec_H v_{i_r}$. By {\em splitting} vertex $u$ we mean: (1) replacing $u$ in ${\cal G}$ with vertices $u_{i}^{1}, \ldots, u_{i}^{r}$ such that the resulting ordering of the vertices in $G$ with respect to $\prec_G$ is $\langle u_1, \ldots, u_{i-1}, u_{i}^{1}, \ldots, u_{i}^{r}, u_{i+1}, \ldots, u_n\rangle$; (2) removing all the edges $e_1, \ldots, e_r$ from ${\cal G}$ and replacing them with the edges $u_{i}^{1}v_{i_r}, u_{i}^{2}v_{i_{r-1}}, \ldots, u_{i}^{r}v_{i_1}$; and (3) replacing every edge $uu_j$ in $G$ with the edges $u_{i}^{s}u_j$, for $s=1, \ldots, r$.
\end{definition}

\begin{figure}
\begin{center}
\begin{minipage}[t]{0.45\linewidth}
\vspace{0pt}
\begin{tikzpicture}
[node/.style={fill,circle,inner sep = 1.5pt},
dot/.style={fill,circle,inner sep = 0.5pt}]

\draw (-3.2,2) node[dot] {};
\draw (-3,2) node[dot] {};
\draw (-2.8,2) node[dot] {};
\draw (-2.4,2) node[node,label=90:{$v_{i_1}$}] (vi1) {};
\draw (-2,2) node[dot] {};
\draw (-1.8,2) node[dot] {};
\draw (-1.6,2) node[dot] {};
\draw (-1.2,2) node[node,label=90:{$v_{i_2}$}] (vi2) {};
\draw (-.8,2) node[dot] {};
\draw (-.6,2) node[dot] {};
\draw (-.4,2) node[dot] {};
\draw (0,2) node[node,label=90:{$v_{i_3}$}] (vi3) {};
\draw (.4,2) node[dot] {};
\draw (.6,2) node[dot] {};
\draw (.8,2) node[dot] {};
\draw (1.2,2) node[node,label=90:{$v_{i_4}$}] (vi4) {};
\draw (1.6,2) node[dot] {};
\draw (1.8,2) node[dot] {};
\draw (2,2) node[dot] {};

\draw (vi1) .. controls +(45:0.4cm) and +(135:0.4cm) .. (vi2);
\draw (vi2) .. controls +(45:0.4cm) and +(135:0.4cm) .. (vi3);

\draw (-2,0) node[dot] {};
\draw (-1.8,0) node[dot] {};
\draw (-1.6,0) node[dot] {};
\draw (-1.2,0) node[node,label=-120:{$u_{j_{1}}$}] (uj1) {};
\draw (-0.8,0) node[dot] {};
\draw (-0.6,0) node[dot] {};
\draw (-0.4,0) node[dot] {};
\draw (0,0) node[node,label=-60:{$u_{j_{2}}$}] (uj2) {};
\draw (0.4,0) node[dot] {};
\draw (0.6,0) node[dot] {};
\draw (0.8,0) node[dot] {};

\draw (uj1) .. controls +(-45:0.4cm) and +(-135:0.4cm) .. (uj2);

\draw (vi1) -- (uj1);
\draw (vi2) -- (uj1);
\draw (vi3) -- (uj1);
\draw (vi2) -- (uj2);
\draw (vi4) -- (uj2);
\end{tikzpicture}
\end{minipage}
\begin{minipage}[t]{0.45\linewidth}
\vspace{0pt}
\begin{tikzpicture}
[node/.style={fill,circle,inner sep = 1.5pt},
dot/.style={fill,circle,inner sep = 0.5pt}]

\draw (-3.2,2) node[dot] {};
\draw (-3,2) node[dot] {};
\draw (-2.8,2) node[dot] {};
\draw (-2.4,2) node[node,label=90:{$v_{i_1}$}] (vi1) {};
\draw (-2,2) node[dot] {};
\draw (-1.8,2) node[dot] {};
\draw (-1.6,2) node[dot] {};
\draw (-1.2,2) node[node,label=90:{$v_{i_2}$}] (vi2) {};
\draw (-.8,2) node[dot] {};
\draw (-.6,2) node[dot] {};
\draw (-.4,2) node[dot] {};
\draw (0,2) node[node,label=90:{$v_{i_3}$}] (vi3) {};
\draw (.4,2) node[dot] {};
\draw (.6,2) node[dot] {};
\draw (.8,2) node[dot] {};
\draw (1.2,2) node[node,label=90:{$v_{i_4}$}] (vi4) {};
\draw (1.6,2) node[dot] {};
\draw (1.8,2) node[dot] {};
\draw (2,2) node[dot] {};

\draw (vi1) .. controls +(45:0.4cm) and +(135:0.4cm) .. (vi2);
\draw (vi2) .. controls +(45:0.4cm) and +(135:0.4cm) .. (vi3);

\draw (-3.2,0) node[dot] {};
\draw (-3,0) node[dot] {};
\draw (-2.8,0) node[dot] {};
\draw (-2.4,0) node[node,label=-90:{$u^1_{j_1}$}] (uj11) {};
\draw (-2,0) node[node] (uj12) {};
\draw (-2.25,0.55) node[] {$u^2_{j_1}$};
\draw (-1.6,0) node[node,label=5:{$u^3_{j_1}$}] (uj13) {};
\draw (-0.6,0) node[dot] {};
\draw (-0.4,0) node[dot] {};
\draw (-0.2,0) node[dot] {};
\draw (0.8,0) node[node,label=175:{$u^1_{j_2}$}] (uj21) {};
\draw (1.2,0) node[node,label=-75:{$u^2_{j_2}$}] (uj22) {};
\draw (1.6,0) node[dot] {};
\draw (1.8,0) node[dot] {};
\draw (2,0) node[dot] {};

\draw (uj11) .. controls +(-45:1.4cm) and +(-135:0.6cm) .. (uj21);
\draw (uj12) .. controls +(-45:1cm) and +(-135:0.6cm) .. (uj21);
\draw (uj13) .. controls +(-45:0.6cm) and +(-135:0.6cm) .. (uj21);
\draw (uj11) .. controls +(-45:2.2cm) and +(-135:0.8cm) .. (uj22);
\draw (uj12) .. controls +(-45:1.8cm) and +(-135:0.8cm) .. (uj22);
\draw (uj13) .. controls +(-45:1.4cm) and +(-135:0.8cm) .. (uj22);

\draw (vi1) -- (uj13);
\draw (vi2) -- (uj12);
\draw (vi3) -- (uj11);
\draw (vi2) -- (uj22);
\draw (vi4) -- (uj21);
\end{tikzpicture}
\end{minipage}
\end{center} \vspace*{-1cm}
\caption{Illustration of the splitting operation when applied to vertices $u_{j_1}$ and $u_{j_2}$.}\vspace*{-0.5cm}
\label{fig:splitting}
\end{figure}
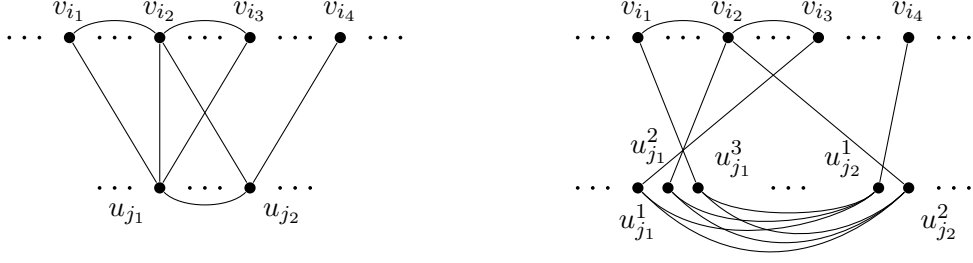

Let ${\cal G}_{split}$ be the graph resulting from ${\cal G}$ by splitting every vertex in $G$ and every vertex in $H$ (in an arbitrary order), where $G_{split}$ is the graph resulting from splitting the vertices in $G$ and $H_{split}$ that resulting from splitting the vertices of $H$. Let $E_{split}$ be the set of edges having one endpoint in $G_{split}$ and the other in $H_{split}$, $L_{split}: G_{split} \longrightarrow 2^{V(H_{split})}$ defined by $L_{split}(u) = \{v \mid uv \in E_{split}\}$ for $u \in V(G_{split})$, and let $\prec_{G_{split}}$ and $\prec_{H_{split}}$ be the orders on $G_{split}$ and $H_{split}$, respectively, resulting from $\prec_G, \prec_H$ after the splitting operation.

\begin{lemma}
\label{lem:splitting}
The graph ${\cal G}_{split}$ satisfies the properties: (i) for every $u \in V(G_{split})$ we have $deg_{G_{split}}(u) \leq \Delta_L \cdot \Delta_G$;\footnote{The degree of a vertex in $H_{split}$ may be unbounded.} (ii) in the graph $(V(G_{split}) \cup V(H_{split}), E_{split})$ every vertex has degree exactly 1 (in particular $|L_{split}(u)| =1$ for every $u \in V(G_{split})$), and (iii) the instance $(G, H, \prec_G, \prec_H, L, k)$ is a yes-instance of p-OLSE if and only if $(G_{split}, H_{split}, \prec_{G_{split}}, \prec_{H_{split}}, L_{split}, k)$ is.
\end{lemma}

\begin{proof}
(i) Consider a vertex $u^r_i$ that is split from vertex $u$. $u^r_i$ is adjacent in $G_{split}$ to every vertex that was split from a neighbor of $u$.  Since $u$ has at most $\Delta_G$ many neighbors in $G_{split}$, and since every neighbor of $u$ was split into at most $\Delta_L$ vertices, it follows that $u^r_i$ has at most $\Delta_L \cdot \Delta_G$ many neighbors in $G_{split}$.

(ii) After splitting a vertex $u \in G$, all resulting vertices are of degree 1 in $(V(G_{split}) \cup V(H_{split}), E_L)$.  Moreover,  splitting the neighbors of $u$ does not change the degree of $u$ in $(V(G_{split}) \cup V(H_{split}), E_L)$.  Therefore, all vertices have degree 1 in $(V(G_{split}) \cup V(H_{split}), E_L)$ after splitting.

(iii) First, we show that if the instance $(G, H, \prec_G, \prec_H, L, k)$ is a yes-instance of p-OLSE then $(G_{split}, H_{split}, \prec_{G_{split}}, \prec_{H_{split}}, L_{split}, k)$ is also a yes-instance.  Let $S$ be a subgraph of $G$ where $|S| = k$ and let $\phi : V(S) \to V(H)$ be a valid embedding.  Let $u_i \in S$ and $v_j \in H$ where $\phi(u_i) = v_j$, without loss of generality, suppose that $u_i$ was split before $v_j$.  After splitting $u_i$ and before splitting $v_j$ there is exactly one vertex $u_i^p$ resulting from splitting $u_i$ such that there is an edge between $u_i^p$ and $v_j$.  After splitting $v_j$ there exists exactly one vertex $v_j^q$ resulting from splitting $v_j$ such that there is an edge $u_i^pv_j^q \in E_{L_{split}}$.  Add $u_j^p$ to the constructed solution $S_{split}$ in  $G_{split}$ and define (the constructed embedding) $\phi_{split}(u_i^p) = v_i^q$.  Given that for any two vertices $u_i, u_j \in S$ and vertices $u_i^r, u_j^s \in S_{split}$ split from $u_i$ and $u_j$ (and similarly for vertices in $H$ and $H_{split}$), $u_i \prec_G u_j$ if and only if $u_i^r \prec_{G_{split}} u_j^s$, and $u_iu_j \in E(G)$ if and only if $u_i^ru_j^s \in E(G_{split})$, it follows that $\phi_{split}$ is a valid embedding from $S_{split}$ into $V(H_{split})$.

Conversely, we show that if $(G_{split}, H_{split}, \prec_{G_{split}}, \prec_{H_{split}}, L_{split}, k)$ is a yes-instance of p-OLSE then so is $(G, H, \prec_G, \prec_H, L, k)$.  Let $S_{split}$ be a subgraph of $G_{split}$ where $|S_{split}| = k$ and let $\phi_{split} : V(S_{split}) \to V(H_{split})$ be a valid embedding.  For each pair of vertices $u_i^p \in S_{split}$ and $v_j^q \in H_{split}$ where $\phi_{split}(u_i^p) = v_j^q$ there are vertices $u_i \in G$ and $v_j \in H$ such that $u_i^p$ and $v_j^q$ were split from $u_i$ and $v_j$, respectively, and $u_iv_j \in E_L$.  Furthermore, for each pair of vertices $u_i^p, u_{i'}^q \in S_{split}$, $u_i^p$ and $u_{i'}^q$ are split from vertices $u_i, u_{i'}$, respectively, where $i \ne i'$.  Thus we can define a subgraph $S$ of $G$ where $u_i \in S$ if any vertex $u_i^p$ split from $u_i$ is in $S_{split}$ and a map $\phi : V(S) \to V(H)$ where $\phi(u_i) = v_j$ if $\phi_{split}(u_i^p) = v_j^q$.  Given that for any two vertices $u_i, u_j \in S$ and vertices $u_i^r, u_j^s \in S_{split}$ split from $u_i$ and $u_j$ (and similarly for vertices in $H$ and $H_{split}$), $u_i \prec_G u_j$ if and only if $u_i^r \prec_{G_{split}} u_j^s$, and $u_iu_j \in E(G)$ if and only if $u_i^ru_j^s \in E(G_{split})$, it follows that $\phi$ is a valid embedding that embeds $S$ into $H$.
\end{proof}

Next, we perform the following operation, denoted {\bf Simplify}, to ${\cal G}_{split}$. Observe that every vertex in $(V(G_{split}) \cup V(H_{split}), E_{split})$ has degree 1 by part (ii) of Lemma~\ref{lem:splitting}. Let $u, u' \in G_{split}$ and let $\{v\}=L_{split}(u)$ and $\{v'\}=L_{split}(u')$. If either (1) $uu' \notin E(G_{split})$ but $vv' \in E(H_{split})$, or (2) both $uu' \in E(G_{split})$ and $vv' \in E(H_{split})$, then we can remove edge $vv'$ from $E(H_{split})$ in case (1) and we can remove both edges $uu'$ and $vv'$ in case (2) without affecting any embedding constraint. Without loss of generality, we will still denote by $(G_{split}, H_{split}, \prec_{G_{split}}, \prec_{H_{split}}, L_{split}, k)$ the resulting instance after the removal of the edges satisfying cases (1) and (2) above. Note that $E(H_{split})= \emptyset$ at this point, and hence if $uu' \in E(G_{split})$ then no valid list embedding can be defined on a subgraph that includes both $u$ and $u'$. Note also that ${\cal G}_{split} - E(G_{split})$ is a realization of a permutation graph $P$ in which the vertices of $G_{split}$ can be arranged on one line according to the order induced by $\prec_{G_{split}}$, and the vertices of $H_{split}$ can be arranged on a parallel line according to the order induced by $\prec_{H_{split}}$. The vertex-set of $P$ corresponds to the edges in $E_{split}$, and two vertices in $P$ are adjacent if and only if their two corresponding edges cross. Note that two vertices in $P$ correspond to two edges of the form $e=uv$ and $e'=u'v'$, where $u, u' \in G_{split}$ and $v, v' \in H_{split}$.
Let ${\cal I}$ be the graph whose vertex-set is $V(P)$ and whose edge set is $E(P) \cup E_c$, where $E_c= \{ee' \mid e, e' \in V(P), e=uv, e'=u'v', uu' \in E(G_{split})\}$ is the set of {\em conflict edges}; that is, ${\cal I}$ consists of the permutation graph $P$ plus the set of conflict edges $E_c$, where each edge in $E_c$ joins two vertices in $P$ whose corresponding endpoints in $G_{split}$ cannot both be part of a valid solution.

\begin{lemma}
\label{lem:boundedconflictdegree}
For every vertex $e \in {\cal I}$, the number of conflict edges incident to $e$ in ${\cal I}$, denoted $deg_c(e)$, is at most $\Delta_L \cdot \Delta_G$.
\end{lemma}

\begin{proof}
By part (i) of Lemma~\ref{lem:splitting}, the degree of a vertex $u \in G_{split}$ is at most $\Delta_L \cdot \Delta_G$.  Note that every conflict edge incident to $e \in {\cal I}$ corresponds to an edge in ${\cal G}$ incident to some $u \in G_{split}$.  Thus for each vertex $e \in {\cal I}$ the vertex $u \in G$ incident to $e$ has at most $\Delta_L \cdot \Delta_G$ neighbors.  The result follows.
\end{proof}

\begin{lemma}
\label{lem:equivalence}
The instance $(G_{split}, H_{split}, \prec_{G_{split}}, \prec_{H_{split}}, L_{split}, k)$, and hence \\ $(G, H, \prec_G, \prec_H, L, k)$, before {\bf Simplify} is applied is a yes-instance of p-OLSE if and only if ${\cal I}$ has an independent set of size $k$.
\end{lemma}

\begin{proof}
A size-$k$ independent set $I$ in ${\cal I}$ corresponds to a set of $k$ edges $u_{i_1}v_{j_1}, \ldots, u_{i_k}v_{j_k}$ in ${\cal G}_{split}$ such that $u_{i_1} \prec_{G_{split}} \ldots \prec_{G_{split}} u_{i_k}$,
$v_{j_1} \prec_{H_{split}} \ldots \prec_{H_{split}} v_{j_k}$, and $S=G_{split}[\{u_{i_1}, \ldots, u_{i_k}\}]$ is a subgraph in $G_{split}$ whose vertices form an independent set. Clearly, the embedding $\phi(u_{i_s}) = \{v_{j_s}\}$, $s=1, \ldots, k$, is a valid embedding that embeds $S$ into $H_{split}$ because it respects both $\prec_{G_{split}}, \prec_{H_{split}}$, and because it respects the embedding constraints. To see why the latter statement is true, note that, for any two vertices $u_{i_s}$ and $u_{i_r}$ ($r \neq s$) in $S$, either there was no edge between $u_{i_s}$ and $u_{i_r}$ before the application of the operation {\bf Simplify}, or there was an edge and got removed by {\bf Simplify}, and in this case there must be also an edge between $v_{j_s}$ and $v_{j_r}$ in $H_{split}$; in either case, $\phi$ respects the embedding constraints.

Conversely, let $\phi$ be a valid embedding that embeds a subgraph $S$ of size $k$ where $V(S)=\{u_{i_1}, \ldots, u_{i_k}\}$, and $\phi(u_{i_s}) = v_{j_s}$, for $s=1, \ldots, k$. We claim that the set of vertices $I=\{e_1=u_{i_1}v_{j_1}, \ldots, e_k=u_{i_k}v_{j_k}\}$ is an independent set in ${\cal I}$. Since $\phi$ is a valid list embedding, no edge in $P$ exists between any two vertices in $I$. Let $e_r=u_{i_r}v_{j_r}$ and $e_s=u_{i_s}v_{j_s}$ be two vertices in $I$, where $r\neq s$. If there is no edge between $u_{i_r}$ and $u_{i_s}$ in $E(G_{split})$, then no edge exists between $e_s$ and $e_r$ in ${\cal I}$. On the other hand, if there is an edge between $u_{i_r}$ and $u_{i_s}$ in $E(G_{split})$, then because $\phi$ is a valid embedding, there must be an edge as well between $v_{i_r}$ and $v_{i_s}$. After applying {\bf Simplify}, the edge between $u_{i_r}$ and $u_{i_s}$ will be removed, and hence no edge exists between $e_s$ and $e_r$ in ${\cal I}$. It follows that $I$ is a size-$k$ independent set in ${\cal I}$.\end{proof}

\begin{lemma}\label{lem:separation}
Let ${\cal C}$ be a hereditary class\footnote{The class ${\cal C}$ is closed under taking subgraphs; that is, every subgraph of a graph in ${\cal C}$ is also in ${\cal C}$.} of graphs on which the {\sc Independent Set} problem is solvable in polynomial time, and let $\Delta \geq 0$ be a fixed integer constant. Let ${\cal C'} =  \{ {\cal I} = (V(P), E(P) \cup E_c) \mid P \in {\cal C} , E_c \subseteq V(P) \times V(P) \}$, where at most $\Delta$ edges in  $E_c$ are incident to any vertex in ${\cal I}$.  Assuming that a graph in ${\cal C'}$ is given as $(V(P), E(P) \cup E_c)$ ($E_c$ is given), the {\sc Independent Set} problem can be solved in $\Ohstar(2^{(\Delta + 1)k}(\Delta (k + 1))^{\Oh(\log(\Delta + 1)k)})$ time on graphs in the class ${\cal C'}$.
\end{lemma}

\begin{proof}
Let $({\cal I} = (V(P), E(P) \cup E_c), k)$ be an instance of {\sc Independent Set}, where ${\cal I} \in {\cal C'}$.
We use the random separation method introduced by Cai et al.~\cite{cai}; this method can be de-randomized in fpt-time using the notion of universal sets and perfect hash functions~\cite{alon95,naor,schmidt}.

We apply the random separation method to the subgraph $(V(P), E_C)$, and color the vertices in $V(P)$ with two colors, ``green" or ``red", randomly and independently. If $I$ is an independent set of size $k$ in ${\cal I}$, since there are at most $\Delta$ edges of $E_c$ that are incident to any vertex in ${\cal I}$, the probability that all vertices in $I$ are colored green and all their neighbors along the edges in $E_c$ are colored red, is at least $2^{-k + \Delta k} = 2^{-(\Delta+1)k}$. If a size-$k$ independent set exists, using universal sets and perfect hash functions, we can find a 2-coloring that will result in the independent set vertices being colored green, and all their neighbors along edges in $E_c$ being colored red in time $\Ohstar(2^{(\Delta + 1)k}(\Delta k + k)^{\Oh(\log(\Delta + 1)k)})$. Therefore, it suffices to determine, given a 2-colored graph ${\cal I}$, whether there is an independent set of size $k$ consisting of green vertices whose neighbors along the edges in $E_c$ are red vertices. We explain how to do so next.

Suppose that the vertices in ${\cal I}$ are colored green or red, and let ${\cal I}_g$ be the subgraph of ${\cal I}$ induced by the green vertices, and ${\cal I}_r$ that induced by the red vertices. Notice that if there is an independent set $I$ consisting of $k$ green vertices whose neighbors along the edges in $E_c$ are red, then for each vertex $u$ in $I$, $u$ is an isolated vertex in the graph $(V({\cal I}_g), E_c)$. Moreover, since $I$ is an independent set, then no edge in $E(P)$ exists between any two vertices in $I$. Therefore, if we form the subgraph $G_0 = (V_0, E_0)$, where $V_0$ is the set of vertices in ${\cal I}_g$ that are isolated with respect to the set of edges $E_c$, and $E_0$ is the set of edges in $E(P)$ whose both endpoints are in $V_0$, then $I$ is an independent set in $G_0$. On the other hand, any independent set of $G_0$ is also an independent set of ${\cal I}$.  Since $G_0$ is a subgraph of $P \in {\cal C}$ and ${\cal C}$ is hereditary, it follows that $G_0 \in {\cal C}$ and we can compute a maximum independent set $I_{max}$ in $G_0$ in polynomial time. If $|I_{max}| \geq k$, then we accept the instance; otherwise, we try the next 2-coloring. If no 2-coloring results in an independent set of size at least $k$, we reject. From the above, given a 2-coloring, clearly we can decide if there is an independent set of size $k$ consisting of green vertices whose neighbors along the edges in $E_c$ are red vertices. The result follows.\end{proof}

\begin{theorem}
\label{thm:mainboundeddegree}
The p-OLSE problem restricted to instances in which $\Delta_G=O(1)$ and $\Delta_L=O(1)$ is $\FPT$.
\end{theorem}

\begin{proof}
Let $(G, H, \prec_G, \prec_H, L, k)$ be an instance of p-OLSE in which both $\Delta_G$ and $\Delta_L$ are upper bounded by a fixed constant. We form the graph ${\cal G}$ and perform the splitting operation described in Definition~\ref{def:splitting} to obtain the instance $(G_{split}, H_{split}, \prec_{G_{split}}, \prec_{H_{split}}, L_{split}, k)$. By Lemma~\ref{lem:splitting}, $(G, H, \prec_G, \prec_H, L, k)$ is a yes-instance of p-OLSE if and only if $(G_{split}, H_{split}, \prec_{G_{split}}, \prec_{H_{split}}, L_{split}, k)$ is. We apply the operation {\bf Simplify} to the instance and construct the graph ${\cal I} = (V(P), E(P) \cup E_c)$ as described above, where $P$ is a permutation graphs. Note that the set of edges $E_c$ is known to us. By Lemma~\ref{lem:equivalence}, $(G, H, \prec_G, \prec_H, L, k)$ is a yes-instance of p-OLSE if and only if ${\cal I}$ has an independent set of size $k$. Observe that we can perform the splitting and the {\bf Simplify} operations, and construct ${\cal I}$ in polynomial time. Since the {\sc Independent Set} problem is solvable in polynomial time on the class of permutation graphs (\eg, see~\cite{kim}), the class of permutation graphs is hereditary, and every vertex in ${\cal I}$ has at most $\Delta_L \cdot \Delta_G$ edges in $E_c$ incident to it by Lemma~\ref{lem:boundedconflictdegree}, it follows from Lemma~\ref{lem:separation} that we can decide if ${\cal I}$ has an independent set of size $k$ in fpt-time, and hence decide the instance $(G, H, \prec_G, \prec_H, L, k)$ in fpt-time. \end{proof}

Unfortunately, the above result does not hold true for p-OLISE, even when $\Delta_G = O(1)$, $\Delta_H=O(1)$, and $\Delta_L = O(1)$, because the number of vertices in $G$ that have the same vertex $v \in H$ in their list may be unbounded.

From Proposition~\ref{prop:whardness}, we know that in the case when $\Delta_L$ is unbounded, and both $\Delta_G = 1$ and $\Delta_H = 1$, p-OLSE is $W[1]$-hard. The following proposition says that the condition $\Delta_H = 1$ is essential for this $W$-hardness result:

\begin{proposition}
\label{prop:randomsimple}
The p-OLSE and p-OLISE problems restricted to instances in which $\Delta_H=0$, $\Delta_G=O(1)$ (resp. $\Delta_G=0$ and $\Delta_H=O(1)$ for p-OLISE by symmetry) and $\Delta_L=\infty$ are $\FPT$.
\end{proposition}

\begin{proof}
We prove the result for p-OLSE. The proof is exactly the same for p-OLISE. The proof uses the random separation method, but is simpler than the proof of Lemma~\ref{lem:separation}. Let $(G, H, \prec_G, \prec_H, L, k)$ be an instance of p-OLSE. Observe that if $S$ is the solution that we are looking for then $V(S)$ must be an independent set since $\Delta_H=0$. Use the random separation method to color $G$ with green or red. Since $\Delta_G = O(1)$, if a solution $S$ exists, then in fpt-time (deterministic) we can find a 2-coloring in which all vertices in $S$ are green and their neighbors in $G$ are red. So we can work under this assumption. Let $G_g$ be the subgraph of $G$ induced by the green vertices, and $G_r$ that induced by the red vertices. Observe that any green vertex in $G_g$ that is not isolated in $G_g$ can be discarded by our assumption (since all neighbors of a vertex in $S$ must be in $G_r$). Therefore, we can assume that $G_g$ is an independent set. We can now compute a maximum cardinality subgraph of $G_g$ that can be (validly) embedded into $H$ using the dynamic programming algorithm in Proposition~\ref{prop:noedges}; if the subgraph has size at least $k$ we accept; otherwise, we try another 2-coloring of $G$. If no 2-coloring of $G$ results in a solution of size at least $k$, we reject.\end{proof}

\section{Parameterization by the Vertex Cover Number}
\label{sec:vcnumber}
In this section we study the parameterized complexity of p-OLSE parameterized by the size of a vertex cover $\nu$ in the graph $G$, shortly (p-VC-OLSE), defined formally as follows:

\paramproblem{} {Two graphs $G$ and $H$ with linear orders $\prec_G$ and $\prec_{H}$ defined on the vertices of $G$ and $H$; a function $L : V(G) \longrightarrow 2^{V(H)}$; and $k \in \nat$}{$\nu$ where $\nu$ is the size of a minimum Vertex Cover}{Is there a subgraph $S$ of $G$ of $k$ vertices and an injective map $\phi: V(S) \longrightarrow V(H)$ such that: (1) $\phi(u) \in L(u)$ for every $u \in S$; (2) for every $u, u' \in S$, if $u \prec_G u'$ then $\phi(u) \prec_H \phi(u')$; and (3) for every $u, u' \in S$, $uu' \in E(G)$ if and only if $\phi(u)\phi(u') \in E(H)$} \\

As noted earlier, the reduction used in Proposition~\ref{prop:whardness} to prove the $W[1]$-hardness of p-OLSE when restricted to instances in which $\Delta_H = 1$, $\Delta_G = 1$ and $\Delta_L=\infty$ results in an instance in which the number of vertices in $G$, and hence $\tau(G)$, is upper bounded by a function of the parameter. Since $\tau(G) < |V(G)|$ it follows immediately that p-VC-OLSE is $W[1]$-hard in this case.  Therefore:

\begin{proposition}
\label{prop:whardnessvc}
p-VC-OLSE restricted to instances in which $\Delta_H = 1$ and $\Delta_G = 1$ is $W[1]$-complete.
\end{proposition}

\vspace*{-1mm}
Therefore, we can focus our attention on studying the complexity of p-VC-OLSE restricted to instances in which $\Delta_L \in O(1)$.
\vspace*{-1mm}
\begin{theorem}
\label{thm:fptvslose}
p-VC-OLSE restricted to instances in which $\Delta_L = O(1)$ can be solved in time $O^*((2\Delta_L)^\nu)$ and hence is $\FPT$.
\end{theorem}
\vspace*{-3mm}
\begin{proof}
Let $\Delta \geq 0$ be any fixed integer, and suppose that $\Delta_L \leq \Delta$. Let $(G, H, \prec_G, \prec_H, L, k, \nu)$ be an instance of p-VC-OLSE, where $k$ is the desired solution size and $\nu$ is the size of a vertex cover in $G$. In fpt-time (in $O^*(1.274^\nu)$ time~\cite{ckxvc}) we can compute a vertex cover $C$ of $G$ of size $\nu$ (if no such vertex cover exists we reject the instance). Let $I = V(G) \setminus C$, and note that $I$ is an independent set of $G$. Suppose that the solution we are seeking (if it exists) is $S$, and the valid mapping of $S$ is $\phi$. Let $S_C = S \cap C$, where $S_C$ is possibly empty, and let $\phi_C$ be the restriction of $\phi$ to $S_C$. We enumerate each subset of $C$ as $S_C$ in time $O^*(2^\nu)$, enumerate each possible mapping from $S_C$ to $L(S_C)$ as $\phi_C$ in time $O^*(\Delta_L^\nu)$, and check the validity of $\phi_C$ in polynomial time (if no valid $\phi_C$ exists, we reject the enumeration); since $\Delta_L \leq \Delta$, the enumeration can be carried out in time $O^*((2\Delta)^{\nu})$. Therefore, we will work under the assumption that the desired solution intersects $C$ at a known subset $S_C$, and that the restriction of $\phi$ to $S_C$ is a known map $\phi_C$, and reject the instance if this assumption is proved to be wrong.

We remove all vertices in $C \setminus S_C$ from $G$ (together with their incident edges) and update $L$ accordingly; without loss of generality, we will still use $G$ to refer to the resulting graph whose vertex-set at this point is $I \cup S_C$. Let $u_{i_1}, \ldots, u_{i_r}$ be the vertices in $S_C$, where $u_{i_1} \prec_G \ldots \prec_G u_{i_r}$. Since $\phi_C$ is valid, we have $\phi(u_{i_1}) \prec_H \ldots \prec_H \phi(u_{i_r})$. We now perform the following operation. For each vertex $u$ in $I$ and each vertex $v \in L(u)$, if setting $\phi(u) = v$ violates the embedding constraint in the sense that either (1) there is a vertex $u_{i_j} \in S_C$ such that $uu_{i_j} \in E(G)$ but $v\phi_C(u_{i_j}) \notin E(H)$ or (2) there is a vertex $u_{i_j} \in S_C$ such that $u \prec u_{i_j}$ (resp. $u_{i_j} \prec u$) but $\phi_C(u_{i_j}) \prec v$ (resp. $v \prec \phi_C(u_{i_j}$)), then remove $v$ from $L(u)$. Afterwards, partition the vertices in $G$ into at most $r+1$ intervals, $I_0, \ldots, I_r$, where $I_0$ consists of the vertices preceding $u_{i_1}$ (with respect to $\preceq_G$), $I_r$ consists of those vertices following $u_{i_r}$, and $I_j$ consists of those vertices that fall strictly between vertices $u_{i_{j-1}}$ and $u_{i_j}$, for $j=1, \ldots, r$. Similarly, partition $H$ into $r+1$ intervals, $I'_0, \ldots, I'_r$, where $I'_0$ consists of the vertices preceding $\phi_C(u_{i_1})$ (with respect to $\prec_H$), $I'_r$ those vertices following $\phi_C(u_{i_r})$, and $I'_j$ those vertices that fall strictly between vertices $\phi_C(u_{i_{j-1}})$ and $\phi_C(u_{i_j})$, for $j=1, \ldots, r$. Clearly, any valid mapping $\phi$ that respects $\prec_G$ and $\prec_H$ must map vertices in the solution that belong to $I_j$ to vertices in $I'_j$, for $j=0, \ldots, r$, in a way that respects the restrictions of $\prec_G$ and $\prec_H$ on $I_j$ and $I'_j$, respectively. On the other hand, since after the above operation every vertex $u$ in $I$ can be validly mapped to any vertex $v \in L(u)$, any injective mapping $\phi_j$ that maps a subset of vertices in $I_j$ to a subset in $I'_j$ in a way that respects the restrictions of $\prec_G$ and $\prec_H$ to $I_j$ and $I'_j$, respectively, can be extended to a valid embedding whose restriction to $S_C$ is $\phi_C$. Therefore, our problem reduces to determining whether there exist injective maps $\phi_j$, $j=0, \ldots, r$, mapping vertices in $I_j$ to vertices in $I'_j$, such that the total number of vertices mapped in the $I_j$'s is $k-r$. Consider the subgraphs $G_j=G[I_j]$ and $H_j=H[I'_j]$, for $j=0, \ldots, r$. Since $G_j$ is an independent set, the presence of edges in $H_j$ does not affect the existence of a valid list mapping from vertices in $G_j$ to $H_j$, and hence those edges can be removed. Therefore, we can solve the opt-OLSE problem using Proposition~\ref{prop:noedges} on the two graphs $G_j$ and $H_j$ to compute a maximum cardinality subset of vertices $S_j$ in $G_j$ that can be validly embedded into $H_j$ via an embedding $\phi_j$. If the union of the $S_j$'s with $S_C$ has cardinality at least $k$, we accept. If after all enumerations of $S_C$ and $\phi_C$ we do not accept, we reject the instance. This algorithm runs in time $O^*((2\Delta_L)^{\nu})$. \end{proof}

\section{Conclusion}\label{sec:conclusion}
We drew a complete complexity landscape of p-OLSE and p-OLISE, and their optimization versions, with respect to the computational frameworks of classical complexity, parameterized complexity, and approximation, in terms of the structural parameters $\Delta_H$, $\Delta_G$ and $\Delta_L$.  We also showed that relaxing the order constraint of p-OLSE and p-OLISE makes the problems significantly easier (in $\Pol$).
There are few interesting open questions that result from our research:

\begin{enumerate}
  \item  What is the complexity of the weighted versions of the problems presented in this paper (vertex-weights or edge-weights), with respect to the frameworks of classical complexity, approximation, and parameterized complexity?
   \item Are there meta-theorems, possibly similar to those for the {\sc Subgraph Isomorphism} and {\sc Graph Embedding} problems (see~\cite{grohebook}), in terms of the structures of $G$ and/or $H$, that would unify some of the complexity results about p-OLSE and p-OLISE presented in the paper?

   \item Finally, graph representations of proteins in practice often fall within the bounds of the parameters assumed for the fixed-parameter algorithms (and approximation algorithms) presented in the paper. How do these algorithms perform in practice?

  \end{enumerate}

\bibliographystyle{plain}
\bibliography{paper}

\begin{thebibliography}{10}

\bibitem{gramm}
J.~Alber, J.~Gramm, J.~Guo, and R.~Niedermeier.
\newblock Computing the similarity of two sequences with nested arc
  annotations.
\newblock {\em Theoretical Computer Science}, 312(2-3):337--358, 2004.

\bibitem{alon95}
N.~Alon, R.~Yuster, and U.~Zwick.
\newblock Color-coding.
\newblock {\em Journal of the ACM}, 42:844--856, 1995.

\bibitem{xiuzhen}
C.~Ashby, X.~Huang, D.~Johnson, I.~Kanj, K.~Walker, and G.~Xia.
\newblock New algorithm for protein structure comparison and classification.
\newblock {\em BMC Genomics}, 14(S2:S1), 2012.

\bibitem{cai}
L.~Cai, S.-M. Chan, and S.-O. Chan.
\newblock Random separation: A new method for solving fixed-cardinality
  optimization problems.
\newblock In {\em IWPEC}, pages 239--250, 2006.

\bibitem{ckxvc}
J.~Chen, I.~Kanj, and G.~Xia.
\newblock Improved upper bounds for vertex cover.
\newblock {\em Theoretical Computer Science}, 411(40-42):3736--3756, 2010.

\bibitem{df2}
R.~Downey and M.~Fellows.
\newblock Parameterized computational feasibility.
\newblock In P.~Clote and J.~Remmel, editors, {\em {Feasible Mathematics II}},
  pages 219--244. Birkhauser, Boston, 1995.

\bibitem{fptbook}
R.~Downey and M.~Fellows.
\newblock {\em Parameterized Complexity}.
\newblock Springer, New York, 1999.

\bibitem{evansphd}
P.~Evans.
\newblock Algorithms and complexity for annotated sequence analysis.
\newblock Technical report, Ph.D. thesis, University of Victoria, 1999.

\bibitem{evans}
P.~Evans.
\newblock Finding common subsequences with arcs and pseudoknots.
\newblock In {\em CPM}, pages 270--280, 1999.

\bibitem{Fagnot2008178}
I.~Fagnot, G.~Lelandais, and S.~Vialette.
\newblock Bounded list injective homomorphism for comparative analysis of
  protein–protein interaction graphs.
\newblock {\em Journal of Discrete Algorithms}, 6(2):178 -- 191, 2008.

\bibitem{Fertin2005}
G.~Fertin, R.~Rizzi, and S.~Vialette.
\newblock Finding exact and maximum occurrences of protein complexes in
  protein-protein interaction graphs.
\newblock In {\em MFCS}, pages 328--339, 2005.

\bibitem{Fertin2009}
G.~Fertin, R.~Rizzi, and S.~Vialette.
\newblock Finding occurrences of protein complexes in protein–protein
  interaction graphs.
\newblock {\em Journal of Discrete Algorithms}, 7(1):90 -- 101, 2009.

\bibitem{grohebook}
J.~Flum and M.~Grohe.
\newblock {\em Parameterized Complexity Theory}.
\newblock Springer-verlag, Berlin, Germany, 2010.

\bibitem{Goldman1999}
D.~Goldman, S.~Istrail, and C.H. Papadimitriou.
\newblock Algorithmic aspects of protein structure similarity.
\newblock In {\em FOCS}, pages 512--522, 1999.

\bibitem{guo}
J.~Gramm, J.~Guo, and R.~Niedermeier.
\newblock On exact and approximation algorithms for distinguishing substring
  selection.
\newblock In {\em FCT}, pages 195--209, 2003.

\bibitem{Hartung2013}
S.~Hartung and R.~Niedermeier.
\newblock Incremental list coloring of graphs, parameterized by conservation.
\newblock pages 258--270, 2010.

\bibitem{hastad97}
J.~H{\aa}stad.
\newblock Clique is hard to approximate within $n^{1-\epsilon}$.
\newblock {\em Electronic Colloquium on Computational Complexity (ECCC)},
  4(38), 1997.

\bibitem{guohuidam}
T.~Jiang, G.~Lin, B.~Ma, and K.~Zhang.
\newblock The longest common subsequence problem for arc-annotated sequences.
\newblock {\em J. Discrete Algorithms}, 2(2):257--270, 2004.

\bibitem{kim}
H.~Kim.
\newblock Finding a maximum independent set in a permutation graph.
\newblock {\em Information Processing Letters}, 36(1):19--23, 1990.

\bibitem{guohui}
G.-H. Lin, Z.-Z. Chen, T.~Jiang, and J.~Wen.
\newblock The longest common subsequence problem for sequences with nested arc
  annotations.
\newblock {\em Journal of Computer and System Sciences}, 65(3):465--480, 2002.

\bibitem{naor}
M.~Naor, L.~J. Schulman, and A.~Srinivasan.
\newblock Splitters and near-optimal derandomization.
\newblock In {\em FOCS}, pages 182--191, 1995.

\bibitem{rolfbook}
R.~Niedermeier.
\newblock {\em Invitation to Fixed-Parameter Algorithms}.
\newblock Oxford University Press, USA, 2006.

\bibitem{py91}
C.~Papadimitriou and M.~Yannakakis.
\newblock Optimization, approximation, and complexity classes.
\newblock {\em Journal of Computer and System Sciences}, 43:425--440, 1991.

\bibitem{schmidt}
J.~Schmidt and A.~Siegel.
\newblock The spatial complexity of oblivious $k$-probe hash functions.
\newblock {\em SIAM Journal on Computing}, 19(5):775--786, 1990.

\bibitem{xiuzhencai}
Y.~Song, C.~Liu, X.~Huang, R.~Malmberg, Y.~Xu, and L.~Cai.
\newblock Efficient parameterized algorithms for biopolymer structuresequence
  alignment.
\newblock {\em IEEE/ACM Trans. Comput. Biology Bioinform.}, 3(4):423--432,
  2006.

\bibitem{west}
D.~West.
\newblock {\em Introduction to graph theory}.
\newblock Prentice Hall Inc., NJ, 1996.

\end{thebibliography}

\end{document}